\documentclass[11pt]{article}

\usepackage{amsmath,amsthm,paralist,bm,yhmath}
\usepackage{amssymb}
\usepackage{epstopdf}
\usepackage{epsfig}
\usepackage{graphicx}
\usepackage{cite}
\usepackage{empheq}
\usepackage[footnotesize]{caption}

\makeatletter
\let\NAT@parse\undefined
\makeatother
\usepackage[sort&compress, numbers]{natbib}

\usepackage{subfigure}
\usepackage{fullpage, upgreek}
\usepackage[usenames]{color}
\newcommand{\Rbb}{\mathbb{R}}
\newcommand{\E}{\mathbb{E}}
\newcommand{\scp}[2]{\langle #1, #2 \rangle}

\newtheorem{theorem}{Theorem}
\newtheorem{definition}{Definition}

\newtheorem{cor}{Corollary}
\newtheorem{lemma}{Lemma}
\newcommand{\inv}[1]{\frac{1}{#1}}
\newcommand{\supp}{{\rm supp}\,}
\newcommand{\tinv}[1]{{\textstyle\frac{1}{#1}}}
\newcommand{\sign}{{\rm sign}\,}

\newcommand{\ud}{\mathrm{d}} 

\newcommand{\ie}{{i.e.}, } 
\newcommand{\eg}{{e.g.}, } 

\DeclareMathOperator*{\argmin}{argmin}

\newcommand{\BeSE}{B$\epsilon$SE}

\newcommand{\cl}{\mathcal}
\newcommand{\bs}{\boldsymbol}
\newcommand{\bb}{\mathbb}
\newcommand{\ang}[2]{d_{S}({#1,#2})}

\newlength{\imgwidth}
\newcommand{\st}{\quad{\rm s.t.}\quad}
\newcommand{\SGR}[2]{\cl N^{#1\times #2}(0,1)}

\newcommand{\sv}{\bar{\boldsymbol z}}

\newcommand{\ys}{\bar{\boldsymbol y}}
\newcommand{\ysc}[1]{\bar{y}_{#1}}

\title{{Robust $1$-Bit Compressive Sensing via\\ Binary Stable Embeddings
  of Sparse Vectors}\footnote{L.\,J. is funded by the Belgian Science
    Policy ``Return Grant'', BELSPO, IAP-VI BCRYPT, and by the Belgian F.R.S-FNRS.
    J.\,L. and R.\,B. were supported by the
    grants NSF CCF-0431150, CCF-0728867, CCF-0926127, CNS-0435425, and
    CNS-0520280, DARPA/ONR N66001-08-1-2065, N66001-11-1-4090, ONR N00014-07-1-0936,
    N00014-08-1-1067, N00014-08-1-1112, and N00014-08-1-1066, AFOSR
    FA9550-07-1-0301 and FA9550-09-1-0432, ARO MURI W911NF-07-1-0185
    and W911NF-09-1-0383, and the Texas Instruments Leadership
    University Program. P.B. is funded by Mitsubishi Electric Research
    Laboratories.}}

\author{Laurent Jacques\thanks{ICTEAM Institute, ELEN Department, Universit\'e catholique de
    Louvain (UCL), B-1348 Louvain-la-Neuve, Belgium. Email:
    laurent.jacques@uclouvain.be}, Jason
  N. Laska\thanks{Dropcam, Inc. San
    Francisco, CA.  Email: jason@dropcam.com} \thanks{Department of Electrical and Computer
    Engineering, Rice University, Houston, TX, 70015 USA. Email: laska@rice.edu, richb@rice.edu}, Petros
  T. Boufounos\thanks{Mitsubishi
    Electric Research Laboratories (MERL) Email: petrosb@merl.com},
  and Richard G. Baraniuk$^{\S}$}
\date{\today}

\begin{document}

\maketitle
\begin{abstract}
  The Compressive Sensing (CS) framework aims to ease the burden on
  analog-to-digital converters (ADCs) by reducing the sampling rate
  required to acquire and stably recover sparse signals.  Practical
  ADCs not only sample but also quantize each measurement to a finite
  number of bits; moreover, there is an inverse relationship between
  the achievable sampling rate and the bit-depth.  In this paper, we
  investigate an alternative CS approach that shifts the emphasis from
  the sampling rate to the number of bits per measurement.  In
  particular, we explore the extreme case of 1-bit CS measurements,
  which capture just their sign.  Our results come in two flavors.
  First, we consider ideal reconstruction from noiseless $1$-bit
  measurements and provide a lower bound on the best achievable
  reconstruction error. We also demonstrate that i.i.d. random
  Gaussian matrices provide measurement mappings that, with
  overwhelming probability, achieve nearly optimal error decay. Next, we
  consider reconstruction robustness to measurement errors and noise
  and introduce the \emph{Binary $\epsilon$-Stable Embedding}
  (B$\epsilon$SE) property, which characterizes the robustness of the 
  measurement process to sign changes.  We show that the same class of
  matrices that provide almost optimal noiseless performance also
  enable such a robust mapping.  On the practical side, we introduce
  the \emph{Binary Iterative Hard Thresholding} (BIHT) algorithm for
  signal reconstruction from 1-bit measurements that offers
  state-of-the-art performance.
\end{abstract}

\section{Introduction}
\label{sec:intro}

Recent advances in signal acquisition theory have led to significant
interest in alternative sampling methods.  Specifically, conventional
sampling systems rely on the Shannon sampling theorem that states that
signals must be sampled uniformly at the Nyquist rate, \ie a rate
twice their bandwidth.  However, the \emph{compressive sensing} (CS)
framework describes how to reconstruct a signal $\bs x \in
\mathbb{R}^{N}$ from the linear measurements
\begin{equation}
\label{eq:ypx}
\bs y = \Phi \bs x,
\end{equation}
where $\Phi \in \mathbb{R}^{M\times N}$ with $M<N$ is an
underdetermined measurement
system~\cite{Can::2006::Compressive-sampling,Don::2006::Compressed-sensing}.
It is possible to design a physical sampling system $\bar{\Phi}$ such
that $\bs y = \Phi\bs x = \bar{\Phi}(x(t))$ where $\bs x$ is a vector
of Nyquist-rate samples of a bandlimited signal $x(t)$, $t \in
\mathbb{R}$.  In this case, (\ref{eq:ypx}) translates to low,
sub-Nyquist sampling rates, providing the framework's axial
significance: CS enables the acquisition and accurate reconstruction of
signals that were previously out of reach, limited by hardware
sampling rates~\cite{TroLasDua::2009::Beyond-Nyquist:} or number of
sensors~\cite{DuaDavTak::2008::Single-pixel-imaging}.  

Although inversion of (\ref{eq:ypx}) seems ill-posed, it has been
demonstrated that $K$-sparse signals, \ie $\bs x \in \Sigma_{K}$
where $\Sigma_{K} := \{\bs x \in \mathbb{R}^{N}: \|\bs x\|_{0} :=
|\mathrm{supp}(\bs x)| \leq K\}$, can be reconstructed
exactly~\cite{Can::2006::Compressive-sampling,Don::2006::Compressed-sensing}.
To do this, we could na\"{i}vely solve for the sparsest signal that
satisfies (\ref{eq:ypx}),
\begin{equation*}
\label{eq:l0}
\bs x^{*}\quad =\quad \argmin_{\bs u \in \mathbb{R}^{N}}\|\bs u\|_{0} \st \bs y = \Phi \bs u;
\tag{\text{R$_{\mathrm{CS}}$}}
\end{equation*}
however, this non-convex program exhibits combinatorial complexity in
the size of the problem
\cite{Natarajan::1995::sparse-approx-lin-syst}. Instead, we solve
\emph{Basis Pursuit} (BP) by relaxing the objective in (\ref{eq:l0})
to the $\ell_{1}$-norm; the result is a convex, polynomial-time
algorithm \cite{Chen::1998::atomic-decomp-BP}. A key realization is
that, under certain conditions on $\Phi$, the BP solution will be
equivalent to that of
(\ref{eq:l0})~\cite{Can::2006::Compressive-sampling}.  This basic
reconstruction framework has been expanded to include numerous fast
algorithms as well as provably robust algorithms for reconstruction
from noisy
measurements~\cite{CandesRIP,CanRomTao::2006::Stable-signal,HalYinZha::2007::A-fixed-point-continuation,Apppro::2007::M.-A.-T.-Figueiredo-and-R.-D.-Nowak,CanTao::2007::The-dantzig-selector:}.
Reconstruction can also be performed with iterative and greedy
methods~\cite{romp,cosamp,BluDav::2008::Iterative-hard}.

Reconstruction guarantees for BP and other algorithms are often
demonstrated for $\Phi$ that are endowed with the \emph{restricted
  isometry property} (RIP), the sufficient condition that the norm of
the measurements is close to the norm of the signal for all sparse
$\bs x$~\cite{CandesRIP}.\footnote{The RIP is in fact not needed to
  demonstrate exact reconstruction guarantees in noiseless settings,
  however it proves quite useful for establishing robust
  reconstruction guarantees in noise.}  This can be expressed, in
general terms, as a \emph{$\delta$-stable embedding}.  Let $\delta \in
(0,1)$ and $\cl X, \cl S \subset \mathbb{R}^{N}$. We say the mapping $\Phi$ is
a \emph{$\delta$-stable embedding} of $\cl X, \cl S$ if
\begin{equation}
\label{eq:se}
(1-\delta)\|\bs x - \bs s\|_{2}^{2} \leq \|\Phi \bs x - \Phi \bs s \|_{2}^{2} \leq (1+\delta)\|\bs x - \bs s\|_{2}^{2},
\end{equation}
for all $\bs x \in \cl X$ and $\bs s \in \cl S$.  The RIP requires that
(\ref{eq:se}) hold for all $\bs x, \bs s \in \Sigma_{K}$; that is, it is a
stable embedding of sparse vectors.
A key result in the CS literature is that, if the coefficients of
$\Phi$ are randomly drawn from a sub-Gaussian distribution, then $\Phi$ will
satisfy the RIP with high probability as long as $M \geq C_{\delta}
K\log(N/K)$, for some constant
$C_{\delta}$~\cite{BarDavDeV::2008::A-Simple-Proof,Dav::2010::Random-observations}.
Several hardware inspired designs with only a few randomized
components have also been shown to satisfy this
property~\cite{Rom::2008::Compressive-sensing,TroLasDua::2009::Beyond-Nyquist:,TroWakDua::2006::Random-filters,Jacques::2009::CMOS-CS-cam}.

In practice, CS measurements must be quantized, \ie each measurement
is mapped from a real value (over a potentially infinite range) to a
discrete value over some finite range.  For example, in uniform
quantization, a measurement is mapped to one of $2^{B}$ distinct
values, where $B$ denotes the number of bits per
measurement. Quantization is an irreversible process that introduces
error in the measurements.  One way to account for quantization error is to treat it as bounded noise and employ robust
reconstruction algorithms. Alternatively, we might try to reduce the
error by choosing the most efficient quantizer for the distribution of
the measurements.  Several reconstruction techniques that specifically
address CS quantization have also been
proposed~\cite{JacHamFad::2009::Dequantizing-compressed,LasBouDav::2009::Demcracy-in-action:,DaiPhaMil::2009::Distortion-rate-functions,ZymBoyCan::2009::Compressed-sensing,SunGoy::2009::Quantization-for-compressed,GraZeo::1971::Quantization-and-Saturation}.

While quantization error is a minor inconvenience, fine quantization
invokes a more burdensome, yet often overlooked source of adversity:
in hardware systems, it is the primary bottleneck limiting sample
rates~\cite{Walden99:ADC,LeRonRee::2005::Analog-to-Digital-Converters}.
In other words, the analog-to-digital converter (ADC) is beholden to the quantizer.  First,
quantization significantly limits the maximum speed of the
 ADC, forcing an exponential decrease
in sampling rate as the number of bits is increased
linearly~\cite{LeRonRee::2005::Analog-to-Digital-Converters}.  Second,
the quantizer is the primary power consumer in an ADC.  Thus, more
bits per measurement directly translates to slower sampling rates and
increased ADC costs. Third, fine quantization is more susceptible to
non-linear distortion in the ADC electronics, requiring explicit
treatment in the
reconstruction~\cite{Bou::2010::Reconstruction-of-sparse}. As we have
seen, the CS framework provides one mechanism to alleviate the quantization bottleneck by reducing the ADC sampling rate.  Is it possible to extend the CS framework to mitigate this problem directly in the quantization
domain by reducing the number of bits per measurement (bit-depth)
instead?

In this paper we concretely answer this question in the affirmative.
We consider an extreme quantization; just one bit per CS measurement,
representing its sign.  The quantizer is thus reduced to a simple
comparator that tests for values above or below zero, enabling
extremely simple, efficient, and fast quantization. A $1$-bit quantizer is also more
robust to a number of commonly encountered non-linear distortions in
the input electronics, as long as they preserve the signs of the
measurements.  

It is not obvious that the signs of the CS measurements retain enough
information for signal reconstruction; for example, it is immediately
clear that the scale (absolute amplitude) of the signal is lost.
Nonetheless, there is strong empirical evidence that signal
reconstruction is
possible~\cite{BouBar::2008::1-Bit-compressive,Bou::2009::Greedy-sparse,LasWenYin::2010::Trust-but-verify:,Bou::2010::Reconstruction-of-sparse}. In
this paper we develop strong theoretical reconstruction and robustness
guarantees, in the same spirit as classical guarantees provided in CS
by the RIP.

We briefly describe the $1$-bit CS framework proposed
in~\cite{BouBar::2008::1-Bit-compressive}.  Measurements of a signal
$\bs x \in \mathbb{R}^{N}$ are computed via
\begin{equation}
\label{eq:defh}
\ys = A(\bs x) := \sign (\Phi \bs x),
\end{equation}
where the sign operator is applied component wise on $\Phi \bs x$,
where $\sign \lambda$ equals $1$ if $\lambda > 0$ and $-1$ otherwise, for
any $\lambda\in\Rbb$. Thus, the measurement operator $A(\cdot)$ is a
mapping from $\mathbb{R}^{N}$ to the Boolean cube\footnote{Generally,
  the $M$-dimensional Boolean cube is defined as $\{0,1\}^M$. Without
  loss of generality, we use $\{-1,1\}^{M}$ instead.}
$\mathcal{B}^{M} := \{-1,1\}^{M}$.  At best, we hope to recover
signals $\bs x \in \Sigma_{K}^{*} := \{\bs x \in S^{N-1}: \|\bs
x\|_{0} \leq K\}$ where $S^{N-1} := \{ \bs x \in \mathbb{R}^{N}: \|\bs
x\|_{2} = 1\}$ is the unit hyper-sphere of dimension $N$.  We restrict
our attention to sparse signals on the unit sphere since, as
previously mentioned, the scale of the signal has been lost during the
quantization process.  To reconstruct, we enforce \emph{consistency}
on the signs of the estimate's measurements, \ie that $A(\bs x^{*}) =
A(\bs x)$.  Specifically, we define a general non-linear
reconstruction decoder $\Delta^{\rm 1bit}(\ys, \Phi, K)$ such
that, for \sloppy{$\bs x^* = \Delta^{\rm 1bit}(\ys, \Phi, K)$},
the solution $\bs x^* $ is
\begin{center}
  \begin{enumerate}
  \item[\emph{(i)}] sparse, \ie satisfies $\|\bs x^*\|_0\leq K=\|\bs x\|_0$,
  \item[\emph{(ii)}] consistent, \ie satisfies $A(\bs x^*) = \ys =
    A(\bs x)$.
  \end{enumerate}
\end{center}
With (\ref{eq:l0}) from CS as a guide, one candidate program for
reconstruction that respects these two conditions is
\begin{equation*}
\label{eq:l0consist}
\bs x^{*}\quad =\quad \argmin_{\bs u \in S^{N-1}}\|\bs u\|_{0} \st \ys
=\sign(\Phi \bs u).
\tag{\text{R$_{\mathrm{1BCS}}$}}
\end{equation*}
Although the parameter $K$ is not explicit in (\ref{eq:l0consist}),
the solution will be $K'$-sparse with $K'\leq K$ because $\bs x$ is a feasible
point of the constraints.

Since (\ref{eq:l0consist}) is computationally intractable,
\cite{BouBar::2008::1-Bit-compressive} proposes a relaxation that
replaces the objective with the $\ell_{1}$-norm and enforces
consistency via a linear convex constraint.  However, the resulting
program remains non-convex due to the unit-sphere requirement.  Be
that as it may, several optimization algorithms have been developed
for the relaxation, as well as a greedy algorithm inspired by the same
ideas~\cite{BouBar::2008::1-Bit-compressive,Bou::2009::Greedy-sparse,LasWenYin::2010::Trust-but-verify:}.
While previous empirical results from these algorithms provide
motivation for the validity of this $1$-bit framework, there have been
few analytical guarantees to date. 

The primary contribution of this paper is a rigorous analysis of the
$1$-bit CS framework.  Specifically, we examine how
  the reconstruction error behaves as we increase the number of
  measurement bits $M$ given the signal dimension $N$ and sparsity
  $K$. We provide two flavors of results.  First, we determine a
lower bound on reconstruction performance from all possible mappings
$A$ with the reconstruction decoder $\Delta^{\rm 1bit}$, \ie the best
achievable performance of this $1$-bit CS framework.  We further
demonstrate that if the elements of $\Phi$ are drawn randomly from
Gaussian distribution or its rows are drawn uniformly from the unit
sphere, then the worst-case reconstruction error
  using $\Delta^{\rm 1bit}$ will decay at a rate almost optimal with
  the number of measurements, up to a log factor in the oversampling
  rate $M/K$ and the signal dimension $N$.  Second, we provide
conditions on $A$ that enable us to characterize the reconstruction
performance even when some of the measurement signs have changed (\eg
due to noise in the measurements). In other words, we derive the
conditions under which robust reconstruction from $1$-bit measurements
can be achieved. We do so by demonstrating that $A$ is a stable
embedding of sparse signals, similar to the RIP. We apply these stable
embedding results to the cases where we have noisy measurements and
signals that are not strictly sparse.  Our guarantees demonstrate that
the $1$-bit CS framework is on sound footing and provide a first step
toward analysis of the relaxed $1$-bit techniques used in practice.

To develop robust reconstruction guarantees, we propose a new tool,
the \emph{binary $\epsilon$-stable embedding} (B$\epsilon$SE), to
characterize $1$-bit CS systems. The B$\epsilon$SE implies that the
normalized angle between any sparse vectors in $S^{N-1}$ is close to
the normalized Hamming distance between their $1$-bit measurements.
We demonstrate that the same class of random $A$ as above exhibit this
property when $M \geq C_{\epsilon}K\log N$ (where $C_{\epsilon}$ is
some constant).  Thus remarkably, there exist $A$ such that the
B$\epsilon$SE holds when both the number of measurements $M$ is
smaller than the dimension of the signal $N$ and the measurement
bit-depth is at minimum. 

As a complement to our theoretical analysis, we introduce a new
$1$-bit CS reconstruction algorithm, \emph{Binary Iterative Hard
  Thresholding} (BIHT).  Via simulations, we demonstrate that BIHT
yields a significant improvement in both reconstruction error as well
as consistency, as compared with previous algorithms.  To gain
intuition about the behavior of BIHT, we explore the way that this
algorithm enforces consistency and compare and contrast it with
previous approaches.  Perhaps more important than the algorithm itself
is the discovery that the BIHT consistency formulation provides a
significantly better feasible solution in noiseless settings, as
compared with previous algorithms.  Finally, we provide a brief
explanation regarding why this new formulation achieves better
solutions, and its connection with results in the machine learning
literature.

Since the first appearance of this work, Plan and Vershynin have
developed additional theoretical results and bounds on the performance of 1-bit
CS, as well as two convex algorithms with theoretical
guarantees~\cite{plan2011dimension,plan2011one,plan2012robust}. The results
in~\cite{plan2011dimension,plan2011one} generalize the B$\epsilon$SE guarantees for
more general classes of signals, including compressible signals
in addition to simply sparse ones. However, the guarantees provided in
that work exhibit worse decay rates in the error performance and the
tightness of the B$\epsilon$SE property. Furthermore, the results of~\cite{plan2011one,plan2012robust} are
intimately tied to reconstruction algorithms, in
contrast to our analysis. We point out similarities and differences
with our results when appropriate in the subsequent development.

In addition to benchmarking the performance of BIHT, our simulations
demonstrate that many of the theoretical predictions that arise from
our analysis (such as the error rate as a function of the number of
measurements or the error rate as a function of measurement Hamming
distance), are actually exhibited in practice. This suggests that our
theoretical analysis is accurately explaining the true behavior of the
framework.

The remainder of this paper is organized as follows.  In
Section~\ref{sec:secmainresultsopt}, we develop performance results for $1$-bit CS in
the noiseless setting.  Specifically we develop a lower bound on
reconstruction performance as well as provide the guarantee that
Gaussian matrices enable this performance.  In
Section~\ref{sec:secmainresults} we introduce the notion of a
B$\epsilon$SE for the mapping $A$ and demonstrate that Gaussian
matrices facilitate this property. We also expand reconstruction
guarantees for measurements with Gaussian noise (prior to
quantization) and non-sparse signals. To make use of these results in
practice, in Section~\ref{sec:introBIHT} we present the BIHT
algorithm for practical $1$-bit reconstruction.  In Section~\ref{sec:sims} we
provide simulations of BIHT to verify our claims.  In
Section~\ref{sec:disc} we conclude with a discussion about
implications and future extensions. To facilitate the flow of the paper and clear
descriptions of the results, most of our proofs are provided in the
appendices.

\section{Noiseless Reconstruction Performance}
\label{sec:secmainresultsopt}

\subsection{Reconstruction performance lower bounds}

\label{subsec:opt}

In this section, we seek to provide guarantees on the reconstruction
error from $1$-bit CS measurements.  Before analyzing this performance
from a specific mapping $A$ with the consistent sparse reconstruction
decoder $\Delta^{\rm 1bit}(\ys, \Phi, K)$, it is instructive to
determine the best achievable performance from measurements acquired
using any mapping. Thus, in this section we seek a lower bound on the
reconstruction error.

We develop the lower bound on the reconstruction error based on how
well the quantizer exploits the available measurement bits. A
distinction we make in this section is that of measurement bits, which
is the number of bits acquired by the measurement system, versus
information bits, which represent the true amount of information
carried in the measurement bits. Our analysis follows similar ideas to
that
in~\cite{Bou::2006::Quantization-and-erasures,BouBar::2007::Quantization-of-sparse},
adapted to sign measurements.

\begin{figure}[t]

  \centering
  \subfigure[\label{fig:orthants}]{
    \includegraphics[height=5cm]{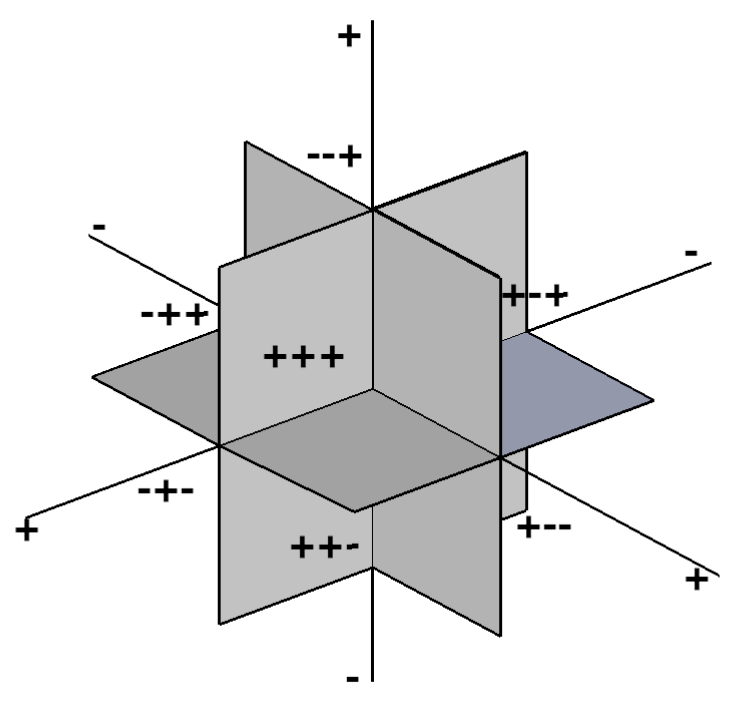}
  }
  \hspace{1in}
  \subfigure[\label{fig:orth_intersect}]
  {
    \includegraphics[height=5cm]{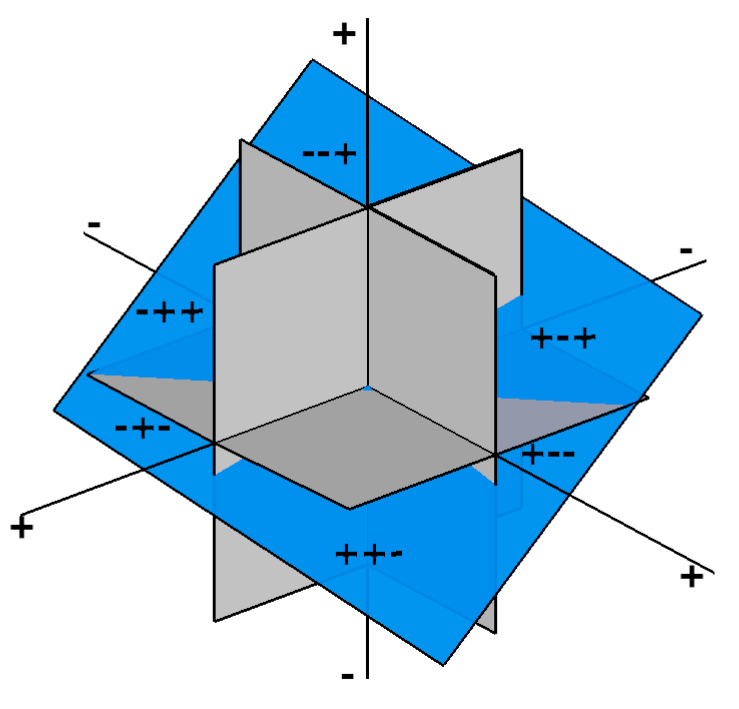}
  }
  \caption{(a) The 8 orthants in $\mathbb{R}^{3}$. (b) Intersection of
    orthants by a 2-dimensional subspace. At most 6 of the 8 available
    orthants are intersected.}
  \label{fig:orthant_geometry}
\end{figure}

We first examine how 1-bit quantization operates on the
measurements. Specifically, we consider the orthants of the
measurement space. An orthant in $\Rbb^M$ is the set of vectors such
that all the vector's coefficients have the same sign pattern
\begin{align}
  \mathcal{O}_{\sv}:=\{\bs z\in\Rbb^M~|~\sign{\bs z}=\sv \},
  \label{eq:orthant}
\end{align}
where $\sv \in \cl B^M$.  Notice that $\cup_{\bar{\bs z}\in \cl B^M}
\mathcal{O}_{\sv}=\Rbb^M$ and $\mathcal{O}_{\sv}\cap
\mathcal{O}_{\sv'}=\emptyset$ if $\sv \neq \sv'$. Therefore, any
$M$-dimensional space is partitioned to $2^M$
orthants. Figure~\ref{fig:orthants} shows the 8 orthants of $\Rbb^3$
as an example. Since 1-bit quantization only preserves the signs of
the measurements, it encodes in which measurement space orthant the
measurements lie. Thus, each available quantization point corresponds
to an orthant in the measurement space. Any unquantized measurement
vector $\Phi\bs x$ that lies in an orthant of the measurement space
will quantize to the corresponding \emph{quantization point} of that
orthant and cannot be distinguished from any other measurement vector
in the same orthant.  To obtain a lower bound on the reconstruction
error, we begin by bounding the number of quantization points (or
equivalently the number of orthants) that are used to encode the
signal.

While there are generally $2^M$ orthants in the measurement space, the
space formed by measuring all sparse signals occupies a small subset
of the available orthants. We determine the number of available
orthants that can be intersected by the measurements in the following
lemma:

\begin{lemma}
  \label{res:intersecions}
  Let $\bs x \in \cl S := \bigcup_{i=1}^L \cl S_i$ belong to a union
  of $L$ subspaces $\cl S_i\subset \mathbb{R}^{N}$ of dimension $K$,
  and let $M\ge 2K$ $1$-bit measurements $\ys$ be acquired via the mapping
  $A: \mathbb{R}^{N} \rightarrow \mathcal{B}^{M}$ as defined in
  (\ref{eq:defh}).  Then the measurements $\ys$ can effectively use at
  most $2^KL\binom{M}{K}$ quantization points, i.e., carry at most
  $K\log_2(2Me/K) + \log_2(L)$ information bits.
\end{lemma}
\begin{proof}
  A $K$-dimensional subspace in an $M$-dimensional space cannot lie in
  all the $2^M$ available octants. For example, as shown in
  Fig.~\ref{fig:orth_intersect}, a 2-dimensional subspace of a
  3-dimensional space can intersect at most 6 of the available
  octants. In Appendix~\ref{app:subspace_intersections}, we
  demonstrate that one arbitrary $K$-dimensional subspace in an
  $M$-dimensional space intersects at most $2^K\binom{M}{K}$ orthants
  of the $2^M$ available. Since $\Phi$ is a linear operator, any
  $K$-dimensional subspace $\cl S_i$ in the signal space
  $\mathbb{R}^N$ is mapped through $\Phi$ to a subspace $\cl S'_i
  =\Phi \cl S_i \subset \mathbb{R}^M$ that is also at most
  $K$-dimensional and therefore follows the same bound.  Thus, if the
  signal of interest belongs in a union $\cl S := \bigcup_{i=1}^L \cl
  S_i$ of $L$ such $K$-dimensional subspaces, then $\Phi \bs x \in \cl
  S' := \bigcup_{i=1}^L \cl S'_i$, and it follows that at most
  $2^KL\binom{M}{K}$ orthants are intersected. This means that at most
  $2^KL\binom{M}{K} \leq L (\frac{2eM}{K})^K$ effective quantization
  points can be used, \ie at most $K\log_2(2eM/K) + \log_2(L)$
  information bits can be obtained.
\end{proof}

Since $K$-sparse signals in any basis $\Psi\in\Rbb^{N\times N}$ belong
to a union of at most $N \choose K$ subspaces in~$\Rbb^N$ with ${N
\choose K} \leq (eN/K)^K$, using
Lemma~\ref{res:intersecions} we can obtain the following
corollary.\footnote{This corollary is easily adaptable to a redundant
  frame $\Psi\in\Rbb^{N\times D}$ with $D\geq N$.}
\begin{cor}
  \label{cor:quantiz-sparsity-subspace}
  Let $\bs x = \Psi\bs\alpha \in \Rbb^N$ be $K$-sparse in a certain
  basis $\Psi\in\Rbb^{N\times N}$, \ie $\bs \alpha \in \Sigma_K$.
  Then the measurements $\ys = A(\bs x)$ can effectively use at
  most $2^K\binom{N}{K}\binom{M}{K}$ 1-bit quantization points, i.e, carry
  at most $2K\log_2(\sqrt 2\,e \sqrt{NM}/K)$ information bits.
\end{cor}

The set of signals of interest to be encoded is the set of unit-norm
$K$-sparse signals $\Sigma^*_K$. Since unit-norm signals of a
$K$-dimensional subspace form a $K$-dimensional unit sphere in that
subspace, $\Sigma^*_K$ is a union of $\binom{N}{K}$ such unit
spheres. The $Q:=2^K\binom{N}{K}\binom{M}{K}$ available quantization
points partition $\Sigma^*_K$ into $Q$ smaller sets, each of which
contains all the signals that quantize to the same point.

To develop the lower bound on the reconstruction error we
examine how $\Sigma^*_K$ can be optimally partitioned with respect to
the worst-case error, given the number of quantization points
used. The measurement and reconstruction process maps each signal in
$\Sigma^*_K$ to a finite set of quantized signals $\cl
Q\subset\Sigma^*_K, |\cl Q|=Q$. At best this map ensures that the
worst case reconstruction error is minimized, \ie
\begin{align}
\label{eq:eopt}
 \epsilon_{\mathrm{opt}}\ =\ 
  \max_{\bs x\in\Sigma^*_K}\,\min_{\bs q\in\cl Q\phantom{\|\!\!}}\|\bs x-\bs q\|_{2}, 
\end{align}
where $\epsilon_{\mathrm{opt}}$ denotes the worst-case quantization
error and $\bs q$ each of the available quantization points. The
optimal lower bound is achieved by designing $\cl Q$ to
minimize~\eqref{eq:eopt} without considering whether the measurement
and reconstruction process actually achieve this design. Thus,
designing the set $\cl Q$ becomes a set covering
problem. Appendix~\ref{app:sph_cap} precises this intuition and proves
the following statement.

\begin{theorem}
  \label{res:lower_bound}
  Let the mapping $A: \mathbb{R}^{N} \rightarrow \mathcal{B}^{M}$ and
  measurements $\ys$ be defined as in (\ref{eq:defh}) and let
  $\bs x \in \Sigma_{K}^{*}$. Then the estimate from the
  reconstruction decoder $\Delta^{\rm 1bit}(\ys,\Phi,K)$ has
  error defined by (\ref{eq:eopt}) of at least
  $$
    \epsilon_{\mathrm{opt}}\ge \frac{K}{2eM + 2K^{3/2}}\ \mathop{\to}_M\ \Omega\left(\frac{K}{M}\right).
  $$
\end{theorem}
Thus, when $M$ is high compared to $K^{3/2}$, the worst-case error
cannot decay at a rate faster than $\Omega(1/M)$ as a function of the
number measurements, no matter what reconstruction algorithm is
used. 

This result assumes noiseless acquisition and provides no guarantees
of robustness and noise resiliency. This is in line with existing
results on scalar quantization in oversampled representations and CS
that state that the distortion due to scalar quantization of noiseless
measurements cannot decrease faster than the inverse of the
measurement
rate~\cite{BouBar::2007::Quantization-of-sparse,SunGoy::20090::Optimal-quantization,ThaVet::2994::Reduction-of-the-MSE-in-R-times,GoyVetTha::1998::Quantized-overcomplete,Bou::2006::Quantization-and-erasures}.

To improve the rate vs.\ distortion trade-off, alternative
quantization methods must be used, such as Sigma-Delta
($\Sigma\Delta$)
quantization~\cite{bib:Aziz96,bib:Thao96b,bib:Candy92,bib:Benedetto06,bib:BoufounosEurasip06,bib:BB_SD_SPIE07,bib:GPSY_SD10}
or non-monotonic scalar
quantization~\cite{Boufounos::2010::univer_rate_effic_scalar_quant}.
Specifically, $\Sigma\Delta$ approaches to CS can achieve error decay
rate of $O((K/M)^{p-1/2})$, where $p$ is the order of the
quantizer~\cite{bib:GPSY_SD10}. However, $\Sigma\Delta$ quantization
requires feedback during the quantization process, which is not
necessary in scalar quantization. Furthermore, the result
in~\cite{bib:GPSY_SD10} only holds for multibit quantizers, not 1-bit
ones. While efficient 1-bit $\Sigma\Delta$ quantization has been shown
for classical sampling~\cite{bib:Thao96b,gunturk2003one}, to the best
of our knowledge, similar results are not currently known for 1-bit
$\Sigma\Delta$ in CS applications. Alternatively, non-monotonic scalar
quantization can achieve error decay exponential in the number of
measurements $M$, even in CS
applications~\cite{Boufounos::2010::univer_rate_effic_scalar_quant}. However,
such a scheme requires a significantly more complex scalar quantizer
and reconstruction approach~\cite{bib:BoufSAMPTA2011}.

Theorem~\ref{res:lower_bound} bounds the best possible performance of
a consistent reconstruction over all possible mappings $A$. However,
not all mappings $A$ will behave as the lower bound suggests.  In the
next section we identify two classes of matrices such that the
mapping~$A$ admits an upper bound on the reconstruction error from a
general decoder~$\Delta^{\rm 1bit}$ that decays almost
optimally.

\subsection{Achievable performance via random projections}
\label{subsec:optgauss}

In this section we describe a class of matrices $\Phi$ such that the
consistent sparse reconstruction decoder $\Delta^{\rm 1bit}(\ys,
\Phi, K)$ can indeed achieve error decay rates of optimal order,
described by Theorem~\ref{res:lower_bound}, with the number of
measurements growing linearly in the sparsity $K$ and logarithmically in
the dimension $N$, as is required in conventional CS. We first focus
our analysis on Gaussian matrices, \ie $\Phi$ such that each element
$\phi_{i,j}$ is randomly drawn i.i.d.\ from the standard Gaussian
distribution, $\mathcal{N}(0,1)$. In the rest of the paper, we use the
short notation $\Phi \sim \SGR{M}{N}$ to characterize such
matrices, and we write $\bs \varphi\sim \SGR{N}{1}$ to describe the
equivalent random vectors in $\Rbb^N$ (\eg the rows of $\Phi$). For
these matrices $\Phi$, we prove the following in
Appendix~\ref{sec:proof-hash}.

\begin{theorem}
\label{thm:hash}
Let $\Phi$ be matrix generated as $\Phi \sim \SGR{M}{N}$, and let the
mapping $A: \mathbb{R}^{N} \rightarrow \mathcal{B}^{M}$ be defined as
in (\ref{eq:defh}).  Fix $0\leq \eta\leq 1$ and $\epsilon_{o}>0$.  If
the number of measurements is
\begin{equation}
\label{eq:gaussian-number-meas-consist}
M \geq \tfrac{2}{\epsilon_{o}}\,\big(2K\,\log(N) + 4K
\log(\tfrac{17}{\epsilon_{o}}) + \log
\tinv{\eta}\big),
\end{equation}
then for all $\bs x,\bs s\in \Sigma^*_K$  we have that
\begin{equation}
\label{eq:radius-implication}
\|\bs x - \bs s\|_{2} > \epsilon_{o}\ \Rightarrow\ A(\bs x) \neq A(\bs
s),  
\end{equation}
or equivalently
$$
A(\bs x) = A(\bs s) \ \Rightarrow\ \|\bs x - \bs s\|_{2} \leq \epsilon_{o},
$$
with probability higher than $1-\eta$.
\end{theorem}

Theorem~\ref{thm:hash} is a uniform reconstruction result, meaning
that with high probability all vectors $\bs x,\bs s\in \Sigma^*_K$ can
be reconstructed as opposed to a non-uniform result where each vector
could be reconstructed with high probability.

As derived in Appendix~\ref{app:proof_epsilon_MK_complex},
Theorem~\ref{thm:hash} demonstrates that if we use Gaussian matrices
in the mapping~$A$, then, given a fixed probability
level $\eta$, the reconstruction decoder $\Delta^{\rm 1bit}(\ys,
\Phi, K)$ will recover signals with error order
\begin{equation*}
  \epsilon_{o}\ =\ O\big( \tfrac{K}{M} \log \tfrac{MN}{K}\big),
\end{equation*}
which decays almost optimally compared to the lower bound given in
Theorem \ref{res:lower_bound} up to a log factor in $MN/K$. Whether
the gap can be closed, with tighter lower or upper bounds is still an
open question. Notice that the hidden proportionality factor in this
last relation depends linearly on $\log 1/\eta$ which is assumed
fixed.

We should also note a few minor extensions of Theorem~\ref{thm:hash}.
We can multiply the rows of $\Phi$ with a positive scalar without
changing the signs of the measurements. By normalizing the rows of the
Gaussian matrix we obtain another class of matrices, ones with rows
drawn uniformly from the unit $\ell_2$ sphere in $\Rbb^N$. It is thus
straightforward to extend the Theorem to matrices with such rows as
well.  Furthermore, these projections are rotation invariant (often
referred to as ``universal'' in CS systems), meaning that the theorem
remains valid for sparse signals in any basis $\Psi$, \ie for $\bs
x,\bs s$ belonging to $\Sigma^*_{\Psi,K}:=\{\bs u=\Psi\bs
\alpha\in\Rbb^N: \bs \alpha \in \Sigma^*_K\}$.  This is true since for
any orthonormal basis $\Psi\in\Rbb^{N \times N}$, $\Phi'=\Phi\Psi \sim
\SGR{M}{N}$ when $\Phi \sim \SGR{M}{N}$.

\subsection{Related Work}
\label{sec:related-works}

A similar result to Theorem~\ref{thm:hash} has been recently shown
for sign measurements of non-sparse signals in the context of
quantization using frame
permutations~\cite{vivekQuantFrame}. Specifically, it has been shown
that reconstruction from sign measurements of signals can be achieved
(almost surely) with an error rate decay arbitrarily close to
$O(1/M)$. Our main contribution here is demonstrating that this result is true
uniformly for all $K$-sparse vectors in $\Rbb^N$, given a \emph{sparse}
and \emph{consistent} decoder. Our results, in addition to introducing
the almost linear dependence on $K$, also show that proving this error
bound uniformly for all $K$-sparse signals involves a logarithmic
penalty in $(MN)/K$. This does not seem to be necessary from the lower
bound in the previous section.  We will see in Section~\ref{sec:sims}
that for Gaussian matrices, the optimal error behavior is empirically
exhibited on average.  Finally, we note that for a constant
$\epsilon_0$, the number of measurements required to guarantee
(\ref{eq:radius-implication}) is $M = O(K\log N)$, nearly the same as
order in conventional~CS.

Furthermore, since the first appearance of our work, a bound on the
achievable reconstruction error for compressible signals and for
signals in arbitrary subsets of $\Rbb^N$ appeared
in~\cite{plan2011dimension,plan2011one}. Specifically for compressible
signals, that works leads to error decay $\epsilon = O((\tfrac{K}{M}
\log \tfrac{N}{K})^{1/4})$, which decreases more slowly (with $K/M$)
than both our bound and the one provided
in~\cite{vivekQuantFrame}. However, the results
in~\cite{plan2011dimension,plan2011one} are for more general classes
of signals.
 
We can also view the binary measurements as a {\em hash} or a {\em
  sketch} of the signal. With this interpretation of the result we
guarantee with high probability that no sparse vectors with Euclidean
distance greater than $\epsilon_{o}$ will ``hash'' to the same binary
measurements.  In fact, similar results play a key role in
\emph{locality sensitive hashing} (LSH), a technique that aims to
efficiently perform approximate nearest neighbors searches from
quantized
projections~\cite{indyk1998approximate,AndDatImm::2006::Locality-sensitive-hashing,AndInd::2008::Near-optimal-hashing,raginsky2009locality}. Most
LSH results examine the performance on point-clouds of a discrete
number of signals instead of the infinite subspaces that we explore in this
paper. Furthermore, the primary goal of the LSH is to preserve the
structure of the nearest neighbors with high probability. Instead, in
this paper we are concerned with the ability to reconstruct the
signal from the hash, as well as the robustness of this reconstruction
to measurement noise and signal model mismatch.  To enable these properties, we require a property of the mapping $A$ that preserves the structure (geometry) of the entire signal set. Thus, in
the next section we seek an embedding property of $A$ that preserves geometry for the set of sparse signals and thus ensures
robust reconstruction.

\section{Acquisition and Reconstruction Robustness}
\label{sec:secmainresults}

\subsection{Binary $\epsilon$-stable embeddings}
\label{subsec:prelim}

In this section we establish an embedding property for the $1$-bit CS mapping $A$
that ensures that the sparse signal geometry is preserved in the
measurements, analogous to the RIP for real-valued measurements. This robustness property enables us to upper
bound the reconstruction performance even when some measurement signs
have been changed due to noise.  Conventional CS achieves robustness
via the $\delta$-stable embeddings of sparse vectors (\ref{eq:se}) discussed
in Section~\ref{sec:intro}.  This embedding is a restricted
\emph{quasi-isometry} between the metric spaces $(\mathbb{R}^{N},
d_{X})$ and $(\mathbb{R}^{M}, d_{Y})$, where the distance metrics
$d_{X}$ and $d_{Y}$ are the $\ell_{2}$-norm in dimensions $N$ and $M$,
respectively, and the domain is restricted to sparse
signals.\footnote{A function $A: X\rightarrow Y$ is called a
  \emph{quasi-isometry} between metric spaces $(X, d_{X})$ and $(Y,
  d_{Y})$ if there exists $C>0$ and $D\geq0$ such that
  $\frac{1}{C}d_{X}(\bs x, \bs s) - D \leq d_{Y}(A(\bs x), A(\bs s))
  \leq Cd_{X}(\bs x, \bs s) + D$ for $\bs x, \bs s \in X$, and $E>0$
  such that $d_{Y}(y,A(\bs x))<E$ for all $y\in Y$
  \cite{Bridson::2008::geom-comb-group-th}.  Since $D=0$ for
  $\delta$-stable embeddings, they are also called bi-Lipschitz
  mappings.}  We seek a similar definition for our embedding;
however, now the signals and measurements lie in the different spaces
$S^{N-1}$ and $\mathcal{B}^{M}$, respectively.  Thus, we first consider appropriate
distance metrics in these spaces.

The {Hamming distance} is the natural distance  for
counting the number of unequal bits between two measurement vectors.
Specifically, for $\bar{\bs a}, \bar{\bs b} \in \mathcal{B}^{M}$ we define the \emph {normalized
Hamming distance} as
\begin{equation*}
d_{H}(\bar{\bs a}, \bar{\bs b}) = \tfrac{1}{M}\sum_{i=1}^{M}\bar a_{i}
\oplus \bar b_{i},
\end{equation*}
where $\bar a \oplus \bar b$ is the XOR operation between $\bar a,
\bar b \in \cl B$ such that $\bar a \oplus \bar b$ equals 0 if $\bar
a=\bar b$ and 1 otherwise.  The distance is normalized such that
$d_{H} \in [0,1]$.  In the signal space we only consider unit-norm
vectors, thus, a natural distance is the angle formed by any two of
these vectors.  Specifically, for $\bs x, \bs s \in S^{N-1}$, we
consider
\begin{equation*}
d_{S}(\bs x, \bs s) := \tfrac{1}{\pi}\arccos\langle \bs x, \bs s \rangle.
\end{equation*} 
As with the Hamming distance, we normalize the true angle
$\arccos\langle \bs x, \bs s \rangle$ such that $d_{S} \in
[0,1]$. Note that since both vectors have the same norm, the inner
product $\langle \bs x, \bs s \rangle$ can easily be mapped to the
$\ell_{2}$-distance using the polarization identity. 

Using these distance metrics we define the binary stable embedding.

\begin{definition}[Binary $\epsilon$-Stable Embedding]
\label{def:bse}
Let $\epsilon \in (0,1)$.  A mapping $A: \mathbb{R}^{N} \rightarrow
\mathcal{B}^{M}$ is a \textbf{binary $\epsilon$-stable embedding}
(B$\epsilon$SE) of order $K$ for sparse vectors if
\begin{equation*}
\ang{\bs x}{\bs s} - \epsilon \leq d_{H}(A(\bs x),A(\bs s)) \leq \ang{\bs x}{\bs s} + \epsilon
\end{equation*}
for all $\bs x,\bs
  s \in S^{N-1}$ with $|\,\supp(\bs x) \cup \supp(\bs s)\,| \leq K$.
\end{definition}
Our definition describes a specific quasi-isometry between the two
metric spaces $(S^{N-1}, d_{S})$ and $(\mathcal{B}^{M}, d_{H})$,
restricted to sparse vectors.  While this mirrors the form of the
$\delta$-stable embedding for sparse vectors, one important difference
is that the sensitivity term $\epsilon$ is additive, rather than
multiplicative, and thus the B$\epsilon$SE is not bi-Lipschitz.  This
is a necessary side-effect of the loss of information due to
quantization.

Any B$\epsilon$SE $A(\cdot)$ of order $2K$ enables robustness
guarantees on any reconstruction algorithm extracting a unit sparse
signal estimate $\bs x^*$ of $\bs x\in\Sigma^*_K$. In this case, the
angular error is immediately bounded by
\begin{equation*}
\ang{\bs x}{\bs x^*}\ \leq d_H\big(A(\bs x),A(\bs x^*)\big) + \epsilon.
\end{equation*}

Thus, if an algorithm returns a unit norm sparse solution with
measurements that are not consistent (\ie $d_{H}(A(\bs x), A(\bs
x^{*})) > 0$), as is the case with several
algorithms~\cite{BouBar::2008::1-Bit-compressive,Bou::2009::Greedy-sparse,LasWenYin::2010::Trust-but-verify:},
then the worst-case angular reconstruction error is close to Hamming
distance between the estimate's measurements' signs and the original
measurements' signs. Section~\ref{sec:sims} verifies this behavior
with simulation results.  Furthermore, in Section~\ref{subsec:corrupt}
we use the \BeSE\ property to guarantee that if measurements are
corrupted by noise or if signals are not exactly sparse, then the
reconstruction error is bounded.

Note that, in the best case, for a B$\epsilon$SE $A(\cdot)$, the
angular error of any sparse and consistent $\Delta^{\rm
  1bit}(\ys,\Phi,K)$ decoder is bounded by $\epsilon$ since then
$d_H\big(A(\bs x),A(\bs x^*)\big) = 0$. As we have seen earlier this
is to be expected because, unlike conventional noiseless CS,
quantization fundamentally introduces uncertainty and exact recovery
cannot be guaranteed. This is an obvious consequence of the mapping of
the infinite set $\Sigma_K^*$ to a discrete set of quantized values.

We next identify a class of
matrices $\Phi$ for which $A$ is a B$\epsilon$SE.

\subsection{Binary $\epsilon$-stable embeddings via random projections}
\label{subsec:results}

As is the case for conventional CS systems with RIP, designing a
$\Phi$ for $1$-bit CS such that $A$ has the B$\epsilon$SE property is
possibly a computationally intractable task (and no such algorithm is
yet known).  Fortunately, an overwhelming number of ``good'' matrices
do exist.  Specifically we again focus our analysis on Gaussian
matrices $\Phi\sim \SGR{M}{N}$ as in as in
Section~\ref{subsec:optgauss}.  As motivation that this choice of
$\Phi$ will indeed enable robustness, we begin with a classical
concentration of measure result for binary measurements from a
Gaussian matrix.
\begin{lemma}
\label{lem:conc-prop-hamming}
Let $\bs x,\bs s\in
S^{N-1}$  be a pair of arbitrary fixed vectors, draw $\Phi$ according to $\Phi\sim\SGR{M}{N}$, and let the
mapping $A: \mathbb{R}^{N} \rightarrow \mathcal{B}^{M}$ be defined as
in (\ref{eq:defh}).  Fix $\epsilon > 0$.  Then we have
\begin{equation}
\label{eq:conc-prop-hamming}
\bb P\left(\ \left|\,d_H\big(A(\bs x),A(\bs s)\big)\ -\ \ang{\bs x}{\bs s}\,\right|\ \leq\ \epsilon\
\right)\ \geq\ 1 - 2\,e^{-2\epsilon^2 M},
\end{equation}
where the probability is with respect to the generation of $\Phi$.
\end{lemma}
\begin{proof}
  This lemma is a simple consequence of Lemma 3.2 in
  \cite{GoeWil::1995::Improved-approximation} which shows that, for one
  measurement, $\bb P[A_j(\bs x) \neq A_j(\bs s)] = \ang{\bs x}{\bs
    s}$. The result then follows by applying Hoeffding's inequality to
  the binomial random variable $M d_H\big(A(\bs x),A(\bs s)\big)$ with
  $M$ trials.
\end{proof}

In words, Lemma~\ref{lem:conc-prop-hamming} implies that the Hamming
distance between two binary measurement vectors $A(\bs x),A(\bs s)$
tends to the angle between the signals $\bs x$ and $\bs s$ as the
number of measurements $M$ increases.  In
\cite{GoeWil::1995::Improved-approximation} this fact is used in the
context of randomized rounding for max-cut problems; however, this
property has also been used in similar contexts as ours with regards
to preservation of inner products from binary
measurements~\cite{ShaJacSta::2010::Randomly-driven,GupRecNow::2010::Sample-complexity}.

The expression (\ref{eq:conc-prop-hamming}) indeed looks similar to
the definition of the B$\epsilon$SE, however, it only holds for a
fixed pair of arbitrary (not necessarily sparse) signals, chosen
prior to drawing $\Phi$.  Our goal is to extend (\ref{eq:conc-prop-hamming}) to cover the entire
set of sparse signals.  Indeed, concentration results similar to
Lemma~\ref{lem:conc-prop-hamming}, although expressed in terms of
norms, have been used to demonstrate the
RIP~\cite{BarDavDeV::2008::A-Simple-Proof}.  These techniques usually
demonstrate that the cardinality of the space of all sparse signals is
sufficiently small, such that the concentration result can be applied
to demonstrate that distances are preserved with relatively few
measurements.

Unfortunately, due to the non-linearity of $A$ we cannot immediately
apply Lemma~\ref{lem:conc-prop-hamming} using the same procedure
as in \cite{BarDavDeV::2008::A-Simple-Proof}. To briefly
summarize, \cite{BarDavDeV::2008::A-Simple-Proof} proceeds by covering
the set of all $K$-sparse signals $\Sigma_K$ with a finite set of
points (with covering radius $\delta>0$).  A concentration inequality
is then applied to this set of points. Since any sparse signal lies in
a $\delta$-neighborhood of at least one such point, the concentration
property can be extended from the finite set to $\Sigma_K$ by bounding
the distance between the measurements of the points within the
$\delta$-neighborhood. Such an approach cannot be used to extend
(\ref{eq:conc-prop-hamming}) to $\Sigma_K$, because the severe
discontinuity of our mapping does not permit us to characterize the
measurements $A(\bs x + \bs s)$ using $A(\bs x)$ and $A(\bs s)$ and
obtain a bound on the distance between measurements of signals in a
$\delta$-neighborhood.

To resolve this issue, we extend Lemma~\ref{lem:conc-prop-hamming}
to include all points within Euclidean balls around the vectors $\bs
x$ and $\bs s$ inside the (sub) sphere $\Sigma^*(T):=\{\bs u\in
S^{N-1}:\supp\bs u \subset T\}$ for some fixed support set $T\subset
[N] := \{1,\,\cdots,N\}$ of size $|T|=D$.  Define the $\delta$-ball
$B_{\delta}(\bs x) := \{ \bs a \in S^{N-1}: \|\bs x - \bs
a\|_{2} < \delta \}$ to be the ball of Euclidean distance $\delta$
around $\bs x$, and let $B_\delta^*(\bs x) := B_\delta(\bs x)\cap\Sigma^*(T)$. 
\begin{lemma}
\label{lem:conc-prop-hamming-balls}
Given $T\subset [N]$ of size $|T|=D$, let $\Phi$ be a
matrix generated as $\Phi\sim\SGR{M}{N}$, and let the mapping $A:
\mathbb{R}^{N} \rightarrow \mathcal{B}^{M}$ be defined as in
(\ref{eq:defh}).  Fix $\epsilon > 0$ and $0\leq\delta\leq1$.  For any
$\bs x,\bs s\in \Sigma^*(T)$, we have
\begin{multline*}
\bb P \left(\ \forall\,\bs u\in B_{\delta}^{*}(\bs x), \forall\,\bs
v\in B_{\delta}^{*}(\bs s),\ \big|\,d_H\big(A(\bs u),A(\bs v)\big)\ -\ \ang{\bs x}{\bs
s}\,\big|\ \leq\ \epsilon + \sqrt{\tfrac{\pi}{2} D}\,\delta\, \right)\ \geq\ 1 - 2\,e^{-2\epsilon^2 M}.
\end{multline*}
\end{lemma}
\noindent The proof of this result is given in
Appendix~\ref{sec:proof-lemma-conc-prop-hamming-balls}. It should be
noted that the proof does not depend on the radial behavior of the
Gaussian pdf in $\Rbb^N$. In other words, this result is easily
generalizable to matrices $\Phi$ whose rows are independent random
vectors drawn from an isotropic pdf .

In words, if the width $\delta$ is sufficiently small, then the Hamming
distance between the $1$-bit measurements $A(\bs u)$, $A(\bs v)$ of any points $\bs u$, $\bs v$ within the
balls  $B_{\delta}^{*}(\bs x)$, $B_{\delta}^{*}(\bs s)$, respectively, will be close to the angle between the centers of the balls.

Lemma~\ref{lem:conc-prop-hamming-balls} is key for providing a similar argument to that
in~\cite{BarDavDeV::2008::A-Simple-Proof}.  
We now simply need to count the number of pairs of $K$-sparse signals
that are euclidean distance $\delta$ apart.  The Lemma
can then be invoked to demonstrate that the angles between all of
these pairs will be approximately preserved by our mapping.\footnote{\label{foot:noiselessproof} We note that the covering argument in the proof of Theorem~\ref{thm:hash} also employs $\delta$-balls in similar fashion but only considers the probability that $d_{H} = 0$, rather than the concentration inequality.}  Thus, with Lemma~\ref{lem:conc-prop-hamming-balls} under our belt, we
demonstrate in Appendix~\ref{sec:proof-thm-1-bit-stable-embed} the
following result.
\begin{theorem}
\label{thm:1-bit-stable-embed}
Let $\Phi$ be a matrix generated as $\Phi\sim\SGR{M}{N}$ and let the
mapping $A: \mathbb{R}^{N} \rightarrow \mathcal{B}^{M}$ be defined as
in (\ref{eq:defh}).  Fix $0\leq \eta\leq 1$ and $\epsilon>0$.  If the
number of measurements is
\begin{equation}
\label{eq:numM}
M\ \geq\ \tfrac{2}{\epsilon^2}\big(K\,\log(N) +
2K\,\log(\tfrac{35}{\epsilon}) + \log(\tfrac{2}{\eta})
\big),
\end{equation}
then with probability exceeding $1-\eta$, the mapping $A$ is a B$\epsilon$SE of order $K$ for sparse vectors.
\end{theorem}
As with Lemma~\ref{lem:conc-prop-hamming-balls}, the theorem extends
easily to matrices $\Phi$ with independent rows in $\Rbb^N$ drawn from
an isotropic pdf in this space.

By choosing $\Phi\sim\SGR{M}{N}$ with $M = O(K\log N)$, with high
probability we ensure that the mapping $A$ is a B$\epsilon$SE.
Additionally, using (\ref{eq:numM}) with a fixed $\eta$ and the
development in Appendix~\ref{app:proof_epsilon_MK_complex}, we find
that the error decreases as

$$
\epsilon\ =\ O\left( \sqrt{\tfrac{K}{M}\, \log \tfrac{MN}{K}} \right).
$$ 

Unfortunately, this decay rate is slower, roughly by a factor of
$\sqrt{K/M}$, than the lower bound in Section~\ref{subsec:opt}.  This
error rate results from an application of the Chernoff-Hoeffding
inequality in the proof of Theorem~\ref{thm:1-bit-stable-embed}.  An
open question is whether it is possible to obtain a tighter bound
(with optimal error rate) for this robustness property.

As with Theorem~\ref{thm:hash}, Gaussian matrices provide a universal
mapping, \ie the result remains valid for sparse signals in a basis
$\Psi\in\Rbb^{N\times N}$. Moreover,
Theorem~\ref{thm:1-bit-stable-embed} can also be extended to rows of
$\Phi$ that are drawn uniformly on the sphere, since the rows of
$\Phi$ in Theorem~\ref{thm:1-bit-stable-embed} can be normalized
without affecting the outcome of the proof.  Note that normalizing the
Gaussian rows of $\Phi$ is as if they had been drawn from a uniform
distribution of unit-norm signals.

We have now established a random construction providing robust
B$\epsilon$SEs with high probability: $1$-bit quantized Gaussian
projections.  We now make use of this robustness by considering an
example where the measurements are corrupted by Gaussian noise.

\subsection{Noisy measurements and compressible signals}
\label{subsec:corrupt}
In practice, hardware systems may be inaccurate when taking measurements; this is often modeled by additive noise.  The mapping $A$ is robust to noise in an unusual way.  After quantization, the measurements can only take the values $-1$ or $1$.  Thus, we can analyze the reconstruction performance from corrupted measurements by considering how many measurements flip their signs.  
For example, we analyze the specific case of Gaussian noise on the measurements prior to quantization, \ie
\begin{equation}
\label{eq:defhn}
A_{n}(\bs x) := \sign(\Phi \bs x + \bs n),
\end{equation}
where $\bs n \in \mathbb{R}^{M}$ has i.i.d.\ elements $n_{i} \sim
\mathcal{N}(0,\sigma^{2})$.  In this case, we demonstrate, via the
following lemma, a bound on the Hamming distance between the corrupted
and ideal measurements with the B$\epsilon$SE from
Theorem~\ref{thm:1-bit-stable-embed} (see
Appendix~\ref{sec:proof-lemma-noise-corruption}).
\begin{lemma}
\label{prop:noise-corruption}
Let $\Phi$ be a matrix generated as $\Phi\sim\SGR{M}{N}$, let the
mapping $A: \mathbb{R}^{N} \rightarrow \mathcal{B}^{M}$ be defined as
in (\ref{eq:defh}), and let $A_{n}: \mathbb{R}^{N} \rightarrow
\mathcal{B}^{M}$ be defined as in (\ref{eq:defhn}).  Let $\bs
n\in\Rbb^M$ be a Gaussian random vector with \emph{i.i.d.}\ components
$n_i\sim \cl N(0,\sigma^2)$. Fix $\gamma > 0$.  Then, given $\bs
x\in\Rbb^N$, we have
\begin{align*}
&\bb E\left(\, d_H\big(A_{\bs n}(\bs x),A(\bs x)\big)\,\right)\ \leq\ e(\sigma,\|\bs x\|_{2}),\\
&\bb P\left(d_H\big(A_{\bs n}(\bs x),A(\bs x)\big) >\ e(\sigma,\|\bs x\|_{2}) +
\gamma\,\right) \leq\ e^{-2M\gamma^2},
\end{align*}
where $e(\sigma,\|\bs x\|_{2}):=\tinv{2} \frac{\sigma}{\sqrt{\|\bs x\|_{2}^2 +
\sigma^2}}\leq \tinv{2}\tfrac{\sigma}{\|\bs x\|_{2}}$.
\end{lemma}
If $\bs x_{n}^{*}$ is the estimate from a sparse consistent
reconstruction decoder $\Delta^{\rm 1bit}(A_{n}(\bs x),\Phi,K)$ from
the measurements $A_{n}(\bs x)$ with $\Phi \sim \cl N^{M\times
  N}(0,1)$ and if $M$ satisfies (\ref{eq:numM}), then it immediately follows from
Lemma~\ref{prop:noise-corruption} and
Theorem~\ref{thm:1-bit-stable-embed} that
\begin{equation}
\label{eq:noiserecon}
\ang{\bs x_{n}^{*}}{\bs x}\ \leq\ d_H\big(A_{\bs n}(\bs x),A(\bs x)\big) + \epsilon\ \leq\ \tinv{2}\tfrac{\sigma}{\|\bs x\|_{2}} +
\gamma + \epsilon,
\end{equation}
with a probability higher than $1-e^{-2M\gamma^2}-\eta$.  Given
alternative noise distributions, \eg Poisson noise, a similar
analysis can be carried out to determine the likely number of sign
flips and thus provide a bound on the error due to noise.

Another practical consideration is that real signals are not always
strictly $K$-sparse.  Indeed, it may be the case that signals are
\emph{compressible}; \ie they can be closely approximated by a
$K$-sparse signal. In this case, we can reuse the non-uniform result
of Lemma~\ref{lem:conc-prop-hamming} to see that, given $\bs x \in
\Rbb^N$ and for $\bs \Phi \sim \cl N^{M \times N}(0,1)$,
$$ 
\bb P\left( d_H\big(A(\bs x), A(\bs x_K)\big) >  d_S(\bs x,\bs x_K) + \gamma \right)\ \leq\ e^{- 2M \gamma^2}.
$$

In similar fashion to (\ref{eq:noiserecon}), if $M$ satisfies
(\ref{eq:numM}), this result and Theorem~\ref{thm:1-bit-stable-embed}
imply that, given $\bs x \in S^{N-1}$ (not necessarily sparse) and for
$\Phi \sim \cl N^{M \times N}(0,1)$ the angular reconstruction error
of $\bs x^*=\Delta^{\rm 1bit}(A(\bs x),\Phi,K)$ is such that $d_S(\bs
x^{*}, \bs x_K) \leq d_H(A(\bs x^*), A(\bs x_K)) + \epsilon =
d_H(A(\bs x), A(\bs x_K)) + \epsilon \leq\ d_S(\bs x,\bs x_K) + \gamma
+ \epsilon$, with probability higher than $1-e^{- 2M
  \gamma^2}-\eta$. Therefore, from the triangular inequality on $d_S$,
this provides the bound
\begin{align*}
d_S(\bs x^{*}, \bs x)\ &\leq\ 2\,d_S(\bs x, \bs x_K)  + \gamma + \epsilon,
\end{align*}
with the same probability. Much like conventional CS results, the
reconstruction error depends on the magnitude of the best $K$-term
approximation error of the signal, here expressed angularly by
$d_S(\bs x, \bs x_K)$.

This reconstruction error bound is non-uniform with respect to the
selection of $\bs x \in \Rbb^N$. A uniform bound on the B$\epsilon$SE for more
general classes of signals is developed
in~\cite{plan2011dimension,plan2011one}, albeit with a worse error
decay---$\epsilon = O((\tfrac{K}{M}\log \tfrac{N}{K})^{1/4})$ for compressible
signals. 

Thus far we have demonstrated a lower bound on the reconstruction
error from $1$-bit measurements (Theorem~\ref{thm:hash}) and
introduced a condition on the mapping $A$ that enables stable
reconstruction in noiseless, noisy, and compressible settings
(Definition~\ref{def:bse}).  We have furthermore demonstrated that a
large class of random matrices---specifically matrices with
coefficients drawn from a Gaussian distribution and matrices with rows
drawn uniformly from the unit sphere---provide good mappings
(Theorem~\ref{thm:1-bit-stable-embed}).  

Using these results we can characterize the error performance of any
algorithm that reconstructs a $K$-sparse signal. If the reconstructed
signal quantizes to the same quantization point as the original data,
then the error is characterized by Theorem~\ref{thm:hash}. If the
algorithm terminates unable to reconstruct a signal consistent with
the quantized data, then Theorem~\ref{thm:1-bit-stable-embed}
describes how far the solution is from the original
signal. Since~\eqref{eq:l0consist} is a combinatorially complex
problem, in the next section we describe a new greedy reconstruction
algorithm that attempts to find a solution as consistent with the
measurements as possible, while guaranteeing this solution is
$K$-sparse.

\section{BIHT: A Simple First-Order Reconstruction Algorithm}
\label{sec:introBIHT}

\subsection{Problem formulation and algorithm definition}
\label{subsec:biht-def}

We now introduce a simple algorithm for the reconstruction of sparse signals
from $1$-bit compressive measurements.  Our algorithm, \emph{Binary
  Iterative Hard Thresholding} (BIHT), is a simple modification of
IHT, the real-valued algorithm from which is takes its
name~\cite{BluDav::2008::Iterative-hard}. Demonstrating theoretical convergence guarantees for BIHT is a subject of future work (and thus not shown in this paper), however the algorithm is of significant value since it \emph{i)} has a simple and intuitive formulation and \emph{ii)} outperforms previous algorithms empirically, demonstrated in Section~\ref{sec:sims}.
We further note that the IHT algorithm has
recently been extended to handle measurement
non-linearities~\cite{Blumensath::2010::cs-nonlin-obs}; however, these
results do not apply to quantized measurements since quantization
does not satisfy the requirements
in~\cite{Blumensath::2010::cs-nonlin-obs}.

We briefly recall that the IHT algorithm consists of two steps that
can be interpreted as follows.  The first step can be thought of as a
gradient descent to reduce the least squares objective $\|\bs y -
\Phi\bs x\|_{2}^{2}/2$.  Thus, at iteration $l$, IHT proceeds by
setting $\bs a^{l+1} = \bs x^{l} + \Phi^{T}(\bs y - \Phi \bs x)$.  The
second step imposes a sparse signal model by projecting $\bs a^{l+1}$
onto the ``$\ell_{0}$ ball'', \ie selecting the $K$ largest in
magnitude elements. Thus, IHT for CS can be thought of as trying to
solve the problem
\begin{equation}
\label{eq:ihtform}
\argmin_{\bs u}\, \tfrac{1}{2}\|\bs y - \Phi \bs u\|_{2}^{2 }\ \st\ \|\bs u\|_0 = K.
\end{equation}

The BIHT algorithm simply modifies the first step of IHT to instead
minimize a consistency-enforcing objective.  Specifically, given an
initial estimate $\bs x^{0} = \bs 0$ and $1$-bit measurements $\ys$, at iteration $l$ BIHT computes
\begin{align}
\label{alg:biht1}
\bs a^{l+1}&= \bs x^{l} + \tfrac{\tau}{2}\Phi^{T}\big(\ys - A(\bs x^{l})\big),\\
\label{alg:biht2}
\bs x^{l+1}&= \eta_{K}(\bs a^{l+1}),
\end{align}
where $A$ is defined as in (\ref{eq:defh}), $\tau$ is a scalar that controls gradient descent step-size, and the function
$\eta_{K}(\bs v)$ computes the best $K$-term approximation of $\bs v$ by thresholding.
Once the algorithm has terminated (either consistency is achieved or
a maximum number of iterations have been reached), we then normalize the final
estimate to project it onto the unit sphere. Section~\ref{subsec:bihtvar} discusses several variations of this algorithm, each with
different properties.

The key to understanding BIHT lies in the formulation of the
objective.  The following Lemma shows that the term
$\Phi^{T}\big(\ys - A(\bs x^{l})\big)$ in (\ref{alg:biht1}) is
in fact the negative subgradient of a convex objective $\cl
J$. Let $[\cdot]_-$ denote the negative function, \ie $([\bs u]_-)_i=[u_i]_-$ with $[u_i]_- = u_i$ if
$u_i <0$ and 0 else, and $\bs u \odot \bs v$ denote the Hadamard product, \ie $(\bs u \odot \bs v)_i = u_iv_i$ for two vectors $\bs u$ and $\bs v$.

\begin{lemma}
\label{lemma:subgradient-l1sided}
The quantity $\frac{1}{2}\,\Phi^{T}\big(A(\bs x) - \ys\big)$
in (\ref{alg:biht1}) is a subgradient of the convex one-sided
$\ell_{1}$-norm
$$
\cl J(\bs x) = \|[\ys
\odot (\Phi\bs x)]_{-}\|_{1},
$$
\end{lemma}
Thus,  BIHT aims to decrease $\cl J$ at each step (\ref{alg:biht1}). 

\begin{proof}
We first note that $\cl J$ is convex. We can write $\cl J(\bs x) =
\sum_{i} \cl J_{i}(\bs x)$ with each convex function $\cl J_i$ given by
\begin{equation*}
\cl J_{i}(\bs x;\ys, \Phi) = 
\begin{cases}
|\langle \bs \varphi_{i},\bs x\rangle|,&\text{if}\ A_i(\bs x)\,\ysc{i} < 0,\\
0,&\text{else},
\end{cases}
\end{equation*}
where $\bs \varphi_{i}$ denotes a row of $\Phi$ and $A_i(\bs
x)=\sign\scp{\bs \varphi_i}{\bs x}$. Moreover, if $\scp{\bs
\varphi_i}{\bs x}\neq 0$, then the gradient of $\cl J_i$~is
$$
\bs\nabla \cl J_{i}(\bs x;\ys, \Phi) =\tinv{2}(A_i(\bs x) -
\ysc{i})\,\bs\varphi_i = \begin{cases}
A_i(\bs x)\, \bs \varphi_{i}&\text{if}\ \ysc{i}\,A_i(\bs x) < 0,\\
0,&\text{else}
\end{cases}\\
$$
while if $\scp{\bs
\varphi_i}{\bs x} = 0$, then the gradient is replaced by the
subdifferential set 
$$
\bs \nabla \cl J_{i}(\bs x;\ys, \Phi) =
\big\{\tfrac{\xi}{2} (A_i(\bs x) - \ysc{i})\,\bs
\varphi_i: \xi \in [0, 1]\big\}\, \ni\, \tfrac{1}{2} (A_i(\bs x) - \ysc{i})\,\bs\varphi_i.
$$
Thus, by summing over $i$ we conclude that
$\frac{1}{2}\,\Phi^{T}\big(A(\bs x) - \ys\big) \in \bs\nabla J(\bs x;\ys, \Phi)$.
\end{proof}

Consequently, the BIHT algorithm can be thought of as
trying to solve the problem:
\begin{equation*}
\bs x^{*}\ =\ \argmin_{\bs u}\ \tau\,\|[\ys \odot (\Phi\bs u)]_{-}\|_{1}    \st\ \|\bs u\|_0 = K,\, \|\bs u\|_{2}=1.
\end{equation*}

Observe that since $\ys \odot (\Phi\bs x)$ simply scales
the elements of $\Phi\bs x$ by the signs $\ys$, 
minimizing the one-sided $\ell_{1}$ objective enforces a positivity
requirement,
\begin{equation}
\label{eq:consistconst}
\ys \odot (\Phi\bs x) \geq 0,
\end{equation}
that, when satisfied, implies consistency.

Previously proposed $1$-bit CS algorithms have used a one-sided
$\ell_{2}$-norm to impose
consistency~\cite{BouBar::2008::1-Bit-compressive,Bou::2009::Greedy-sparse,LasWenYin::2010::Trust-but-verify:,Bou::2010::Reconstruction-of-sparse}.
Specifically, they have applied a constraint or objective that takes
the form $\|[\ys \odot (\Phi\bs x)]_{-}\|_{2}^{2}/2$.  Both the
one-sided $\ell_{1}$ and $\ell_{2}$ functions imply a consistent
solution when they evaluate to zero, and thus, both approaches are
capable of enforcing consistency. However, the choice of the
$\ell_{1}$ vs.\ $\ell_2$ penalty term makes a significant difference
in performance depending on the noise conditions. We explore this
difference in the experiments in Section~\ref{sec:sims}.

\subsection{BIHT shifts}
\label{subsec:bihtvar}

Several modifications can be made to the BIHT algorithm that may
improve certain performance aspects, such as consistency,
reconstruction error, or convergence speed.  While a comprehensive
comparison is beyond the scope of this paper, we believe that such
variations exhibit interesting and useful properties that should be
mentioned. 

\textbf{Projection onto sphere at each iteration.}  We can enforce
that every intermediate solution have unit $\ell_{2}$ norm.  To do
this, we modify the ``impose signal model'' step (\ref{alg:biht2}) by
normalizing after choosing the best $K$-term approximation, \ie we
replace the update of $\bs x^{l+1}$ in (\ref{alg:biht2}) by $\bs
x^{l+1} = U(\eta_{K}(\bs a^{l+1}))$, where $U(\bs v) = \bs
v/\|\bs v\|_{2}$. While this step is found in previous algorithms such
as~\cite{BouBar::2008::1-Bit-compressive,Bou::2009::Greedy-sparse,LasWenYin::2010::Trust-but-verify:},
empirical observations suggest that it not required for BIHT to
converge to an appropriate solution.

If we choose to impose the projection, $\Phi$ must be appropriately
normalized or, equivalently, the step size of the gradient descent
must be carefully chosen. Otherwise, the algorithm will not converge.
Empirically, we have found that for a Gaussian matrix, an appropriate scaling is
$1/(\sqrt{M}\|\Phi\|_{2})$, where the $1/\|\Phi\|_{2}$ controls the
amplification of the estimate from $\Phi^{T}$ in the gradient descent
step (\ref{alg:biht1}) and the $1/\sqrt{M}$ ensures that $\|\ys
- A(\bs x^{l})\|_{2} \leq 2$.  Similar gradient step scaling
requirements have been imposed in the conventional IHT algorithm and other
sparse recovery algorithms as
well (\eg \cite{HalYinZha::2007::A-fixed-point-continuation}).

\textbf{Minimizing hinge loss.}  The one-sided $\ell_{1}$-norm is
related to the \emph{hinge-loss} function in the machine learning
literature, which is known for its robustness to
outliers~\cite{HasTibFri::2009::The-elements-of-statistical}.  Binary
classification algorithms seek to enforce the same consistency
function as in (\ref{eq:consistconst}) by minimizing a function
$\sum_i [\kappa - \ysc{i} (\Phi\bs x)_i ]_{+} = \|[\kappa \bs 1 - \ys
\odot (\Phi\bs x) ]_{+}\|_1$, where $[\cdot]_{+}$ sets negative
elements to zero.  When $\kappa>0$, the objective is both convex and
has a non-trivial solution. Further connections and interpretations
are discussed in Section~\ref{sec:sims}. Thus, rather than minimizing
the one-sided $\ell_{1}$ norm, we can instead minimize the hinge-loss.
The hinge-loss can be interpreted as ensuring that the minimum value
that an unquantized measurement $(\Phi\bs x)_{i}$ can take is bounded
away form zero, \ie $|(\Phi\bs x)_{i}| \geq \kappa$.  This requirement
is similar to the sphere constraint in that it avoids a trivial
solution; however, will perform differently than the sphere
constraint.  In this case, in the gradient descent step
(\ref{alg:biht1}), we instead compute
$$ \bs a^{l+1} = \bs x^{l} - \tau\Theta^{T}(\mathrm{sign}(\Theta\bs x^{l}
- \kappa)-1)/2
$$
where $\Theta = (\ys\odot\Phi)$ scales the rows of $\Phi$ by the
signs of $\ys$. Again, the step size must be chosen
appropriately, this time as $C_{\kappa}/\|\Phi\|_{2}$, where $C_{\kappa}$ is a parameter that depends on $\kappa$.

\textbf{General one-sided objectives.} In general, any function $\cl
R(\bs x) = \sum \cl R_{i}(x_{i})$, where $\cl R_{i}$ is continuous and
has a negative gradient for $x_{i} \leq 0$ and is $0$ for $x_{i} > 0$,
can be used to enforce consistency.  To employ such functions, we
simply compute the gradient of $\cl R$ and apply it in
(\ref{alg:biht1}). As an example, the previously mentioned one-sided
$\ell_{2}$-norm has been used to enforce consistency in several
algorithms.  We can use it in BIHT by computing
$$
\bs a^{l} = \bs x^{l} + \tau\Phi^{T}[\ys \odot \Phi\bs x^{l}]_{+}
$$ in (\ref{alg:biht1}).  We compare and contrast the behavior of the
one-sided $\ell_{1}$ and $\ell_{2}$ norms in
Section~\ref{sec:sims}. 

As another example, in similar fashion to the Huber
norm~\cite{Hub::1973::Robust-regression:}, we can combine the
$\ell_{1}$ and $\ell_{2}$ functions in a piecewise fashion.  One
potentially useful objective is $\sum \cl R_{i}(\bs x)$, where $\cl
R_{i}$ is defined as follows:
\begin{equation}
\cl R_{i}(\bs x) = \left\{ 
 \begin{array}{ll}
 0,& x_{i} \geq 0,\\
| x_{i}|,& -\frac{1}{2} \leq x_{i} < 0,\\
 x_{i}^{2} + \frac{1}{4},&  x_{i} < -\frac{1}{2}.
 \end{array}
 \right.
\end{equation}
While similar, this is not exactly a one-sided Huber norm.  In a one-sided
Huber-norm, the square ($\ell_{2}$) term would be applied to values
near zero and the magnitude ($\ell_{1}$) term would be applied to
values significantly less than zero, the reverse of what we propose
here.

This hybrid objective can provide different robustness properties or
convergence rates than the previously mentioned objectives.
Specifically, during each iteration it may allow us to take advantage
of the shallow gradient of the one-sided $\ell_{2}$ cost for large
numbers of measurement sign discrepancies and the steeper gradient of
the one-sided $\ell_{1}$ cost when most measurements have the correct
sign. This objective can be applied in BIHT as with the other
objectives, by computing its gradient and plugging it into
(\ref{alg:biht1}).

\section{Experiments}
\label{sec:sims}

In this section we explore the performance of the BIHT algorithm and compare it to the performance of previous
$1$-bit CS algorithms.  To make the comparison as straightforward as possible, we
reproduced the experiments
of~\cite{LasWenYin::2010::Trust-but-verify:} with the BIHT
algorithm. 

The experimental setup is as follows.  For each data point, we draw a
length-$N$, $K$-sparse signal with the non-zero entries drawn
uniformly at random on the unit sphere, and we draw a new $M\times N$
matrix $\Phi$ with each entry $\phi_{ij} \sim \mathcal{N}(0,1)$.  We
then compute the binary measurements $\ys$ according to
(\ref{eq:defh}).  Reconstruction of $\bs x^{*}$ is performed from $\ys$ with three algorithms: \emph{matching sign pursuit}
(MSP)~\cite{Bou::2009::Greedy-sparse}, \emph{restricted-step
  shrinkage} (RSS)~\cite{LasWenYin::2010::Trust-but-verify:}, and BIHT
(this paper); the algorithms will be depicted by dashed, dotted, and
triangle lines, respectively.  Each reconstruction in this setup is
repeated for 1000 trials and with a fixed $N=1000$ and $K=10$ unless
otherwise noted.  Furthermore, we perform the trials for $M/N$ within
the range $[0,2]$.  Note that when $M/N > 1$, we are acquiring more
measurements than the ambient dimension of the signal.  While the $M/N
> 1$ regime is not interesting in conventional CS, it may be very
practical in $1$-bit systems that can acquire sign measurements at
extremely high, super-Nyquist rates.

\begin{figure}[!h]
   \centering
   \subfigure[]{\includegraphics[width=.9\imgwidth]{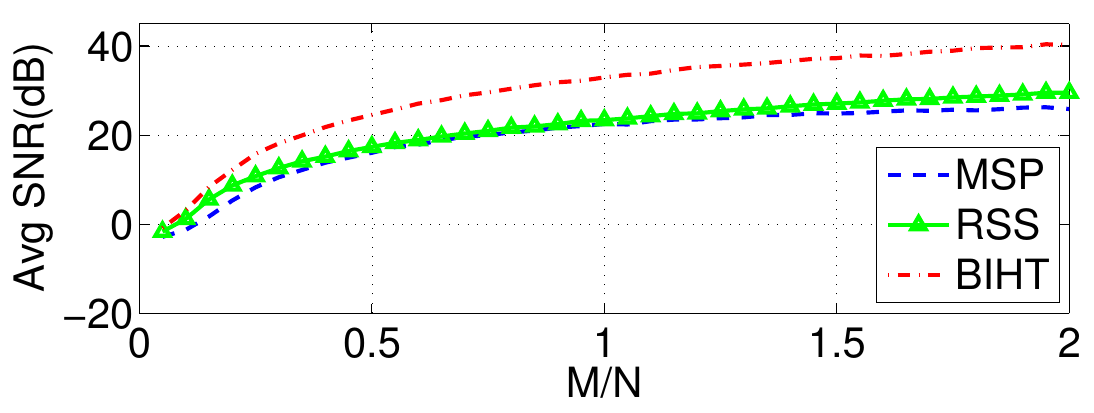}}
   \subfigure[]{\includegraphics[width=.9\imgwidth]{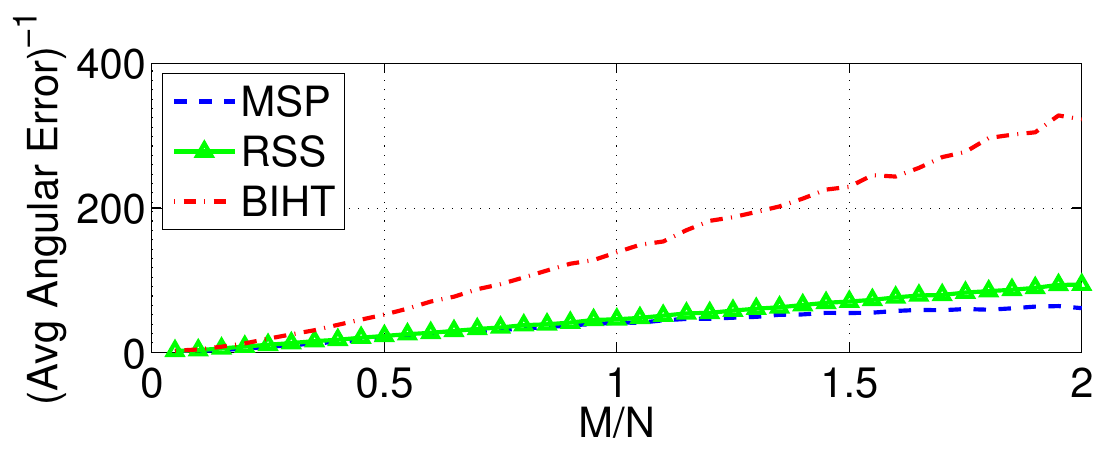}}
   \caption{Average reconstruction angular error
       $\epsilon_{\mathrm{sim}}$ vs. $M/N$, plotted two ways. (a) SNR
       in decibels, and (b) Inverse angular error
       $\epsilon_{\mathrm{sim}}^{-1}$. The plot demonstrates that BIHT
       yields a considerable improvement in reconstruction error,
       achieving an SNR as high as $40$dB when $M/N = 2$.
       Furthermore, we see that the error behaves according
       $\epsilon^{-1}_{\mathrm{sim}} = O(M)$, implying that on
     average we achieve the optimal performance rate given in
     Theorem~\ref{res:lower_bound}.}
   \label{fig:avgerror}
\end{figure}
\textbf{Average error.} We begin by measuring the average
reconstruction angular error $\epsilon_{\mathrm{sim}} := d_{S}(\bs x,
\bs x^{*})$ over the $1000$ trials.  Figure~\ref{fig:avgerror}
displays the results of this experiment in two different ways:
\emph{(i)} the signal-to-noise ratio (SNR)\footnote{In this paper we
  define the reconstruction SNR in decibels as $\mathrm{SNR}(\bs x) :=
  10\log_{10}(\|\bs x\|_{2}^{2}/\|\bs x-\bs x^{*}\|_{2}^{2})$. Note
  that this metric uses the standard euclidean error and not angular
  error.} in Figure~\ref{fig:avgerror}(a), to demonstrate that the
performance of these techniques is practical (since the angular error
is unintuitive to most observers), and \emph{(ii)} the inverse of the
angular error, \ie $\epsilon_{\mathrm{sim}}^{-1}$ in
Figure~\ref{fig:avgerror}(b), to compare with the performance
predicted by Theorem~\ref{thm:hash}.

We begin by comparing the performance of the algorithms.  While we can
observe that the angular error of each algorithm follows the
same trend, BIHT obtains smaller error (or higher SNR) than the
others, significantly so when $M/N$ is greater than $0.35$.  The
discrepancy in performance could be due to difference in the
algorithms themselves, or perhaps, differences in their formulations
for enforcing consistency.  This is explored later in this section.

We now consider the actual performance trend.  We see from
Figure~\ref{fig:avgerror}(b) that, above $M/N = 0.35$ each line
appears fairly linear, albeit with a different slope, implying that
with all other variables fixed, $\epsilon_{\mathrm{sim}} =
O(1/M)$.  This is on the order of the optimal performance
as given by the bound given in Theorem~\ref{res:lower_bound} and
predicted by Theorem~\ref{thm:hash} for Gaussian matrices.

\begin{figure}[t]
   \centering
   \subfigure[\label{fig:AngVsHamm-01}M/N = 0.1]{\includegraphics[width=.305\textwidth]{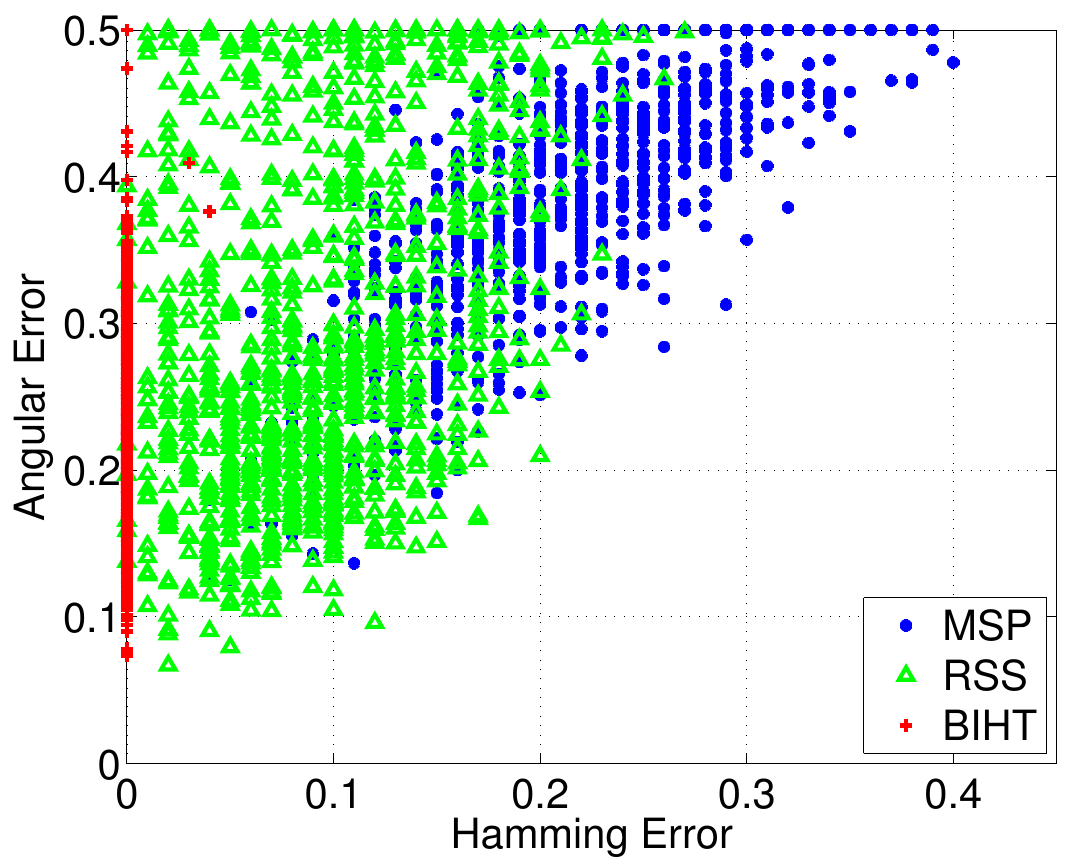}}
   \subfigure[\label{fig:AngVsHamm-07}M/N = 0.7]{\includegraphics[width=.32\textwidth]{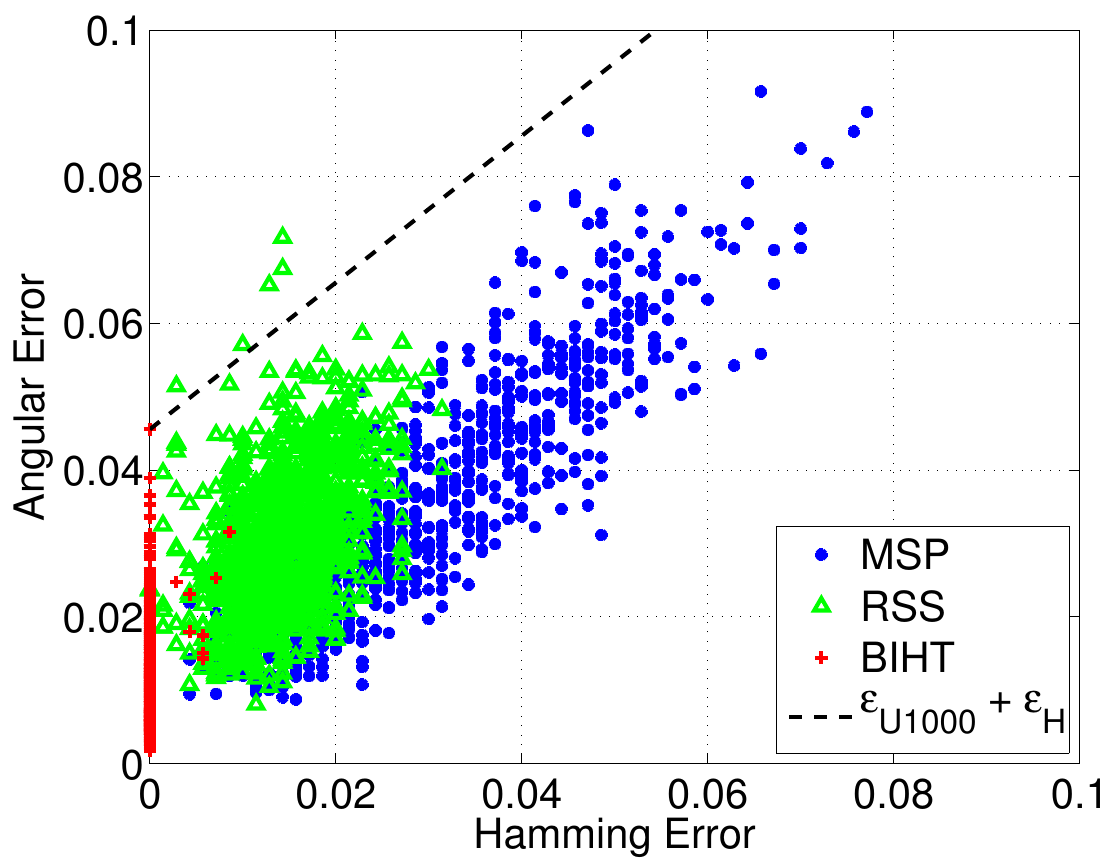}}
   \subfigure[\label{fig:AngVsHamm-15}M/N =1.5]{\includegraphics[width=.32\textwidth]{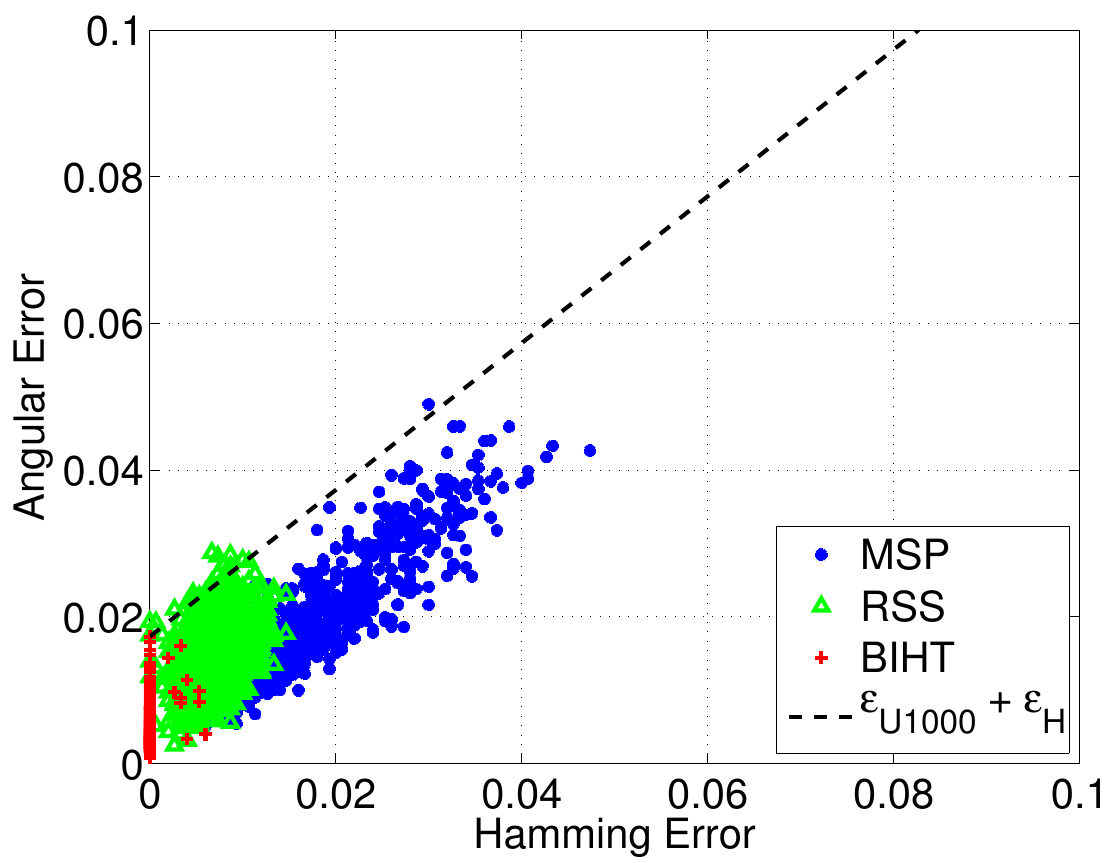}}
   \caption{Reconstruction angular error $\epsilon_{\mathrm{sim}}$
     vs.\ measurement Hamming error $\epsilon_{H}$.  BIHT returns a
     consistent solution in most trials.  For sufficiently large $M/N$
     regimes, we see a linear relationship $\epsilon_{\mathrm{sim}}
     \approx C + \epsilon_{H}$ between the average angular error
     $\epsilon_{\mathrm{sim}}$ and the hamming error $\epsilon_{H}$
     where $C$ is constant (see (a) and (b)).  The B$\epsilon$SE
     formulation in Definition~\ref{def:bse} predicts that the angular
     error is bounded by the hamming error $\epsilon_{H}$ in addition
     to an offset $\epsilon$. The dashed line
     $\epsilon_{U1000}+\epsilon_{H}$ denotes the empirical upper bound
     for $1000$ trials.\label{fig:AngVsHamm}}
\end{figure}

\textbf{Consistency.}  We also expose the relationship between the
Hamming distance $d_{H}(A(\bs x),A(\bs x^{*}))$ between the
measurements of the true and reconstructed signal and the angular
error of the true and reconstructed signal.
Figure~\ref{fig:AngVsHamm} depicts the Hamming distance vs.\ angular
error for three different values of $M/N$.  The particularly striking
result is that BIHT returns significantly more consistent
reconstructions than the two other algorithms. We observed
  this effect for ratio $M/N$ as small as $M/N = 0.1$
  (Figure~\ref{fig:AngVsHamm-01}).  This is clear from the fact that most of the red (plus)
points lie on the y-axis while the majority of blue (dot) or green
(triangle) points do not.  We find that, even in significantly
``under-sampled'' regimes like $M/N=0.1$, where the B$\epsilon$SE is
unlikely to hold, BIHT is likely to return a consistent solution
(albeit with high variance of angular errors).  We also find that in
``over-sampled'' regimes such as $M/N = 1.7$, the range of angular
errors on the y-axis is small.  Indeed, the range of angular errors
shrinks as $M/N$ increases, implying an imperical tightening of the
B$\epsilon$SE upper and lower bounds.

We can infer an interesting performance trend from
Figures~\ref{fig:AngVsHamm}(b) and (c), where the B$\epsilon$SE
property may hold.  Since the RSS and MSP algorithms often do not
return a consistent solution, we can visualize the relationship between
angular error and hamming error.  Specifically, on average the angular
reconstruction error is a linear function of hamming error,
$\epsilon_{H} = d_{H}(A(\bs x),A(\bs x^{*}))$, as similarly expressed by the
reconstruction error bound provided by B$\epsilon$SE.  Furthermore, if
we let $\epsilon_{1000}$ be the largest angular error (with consistent
measurements) over $1000$ trials, then we can suggest an empirical
upper bound for BIHT of $\epsilon_{1000} + \epsilon_{H}$.  This upper bound is denoted by the dashed line in
Figures~\ref{fig:AngVsHamm}(b) and (c).

\begin{figure}[t]
   \centering
   \subfigure[]{\includegraphics[width=.9\imgwidth]{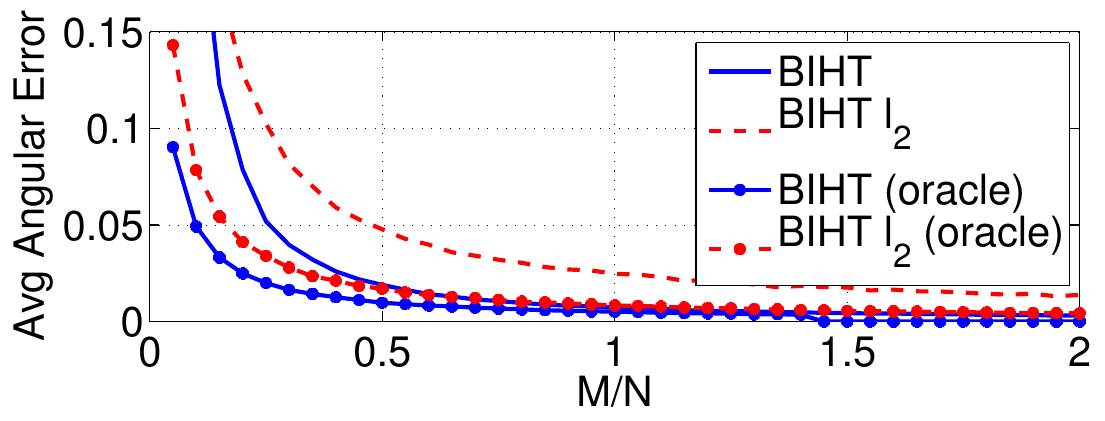}}
   \subfigure[]{\includegraphics[width=.9\imgwidth]{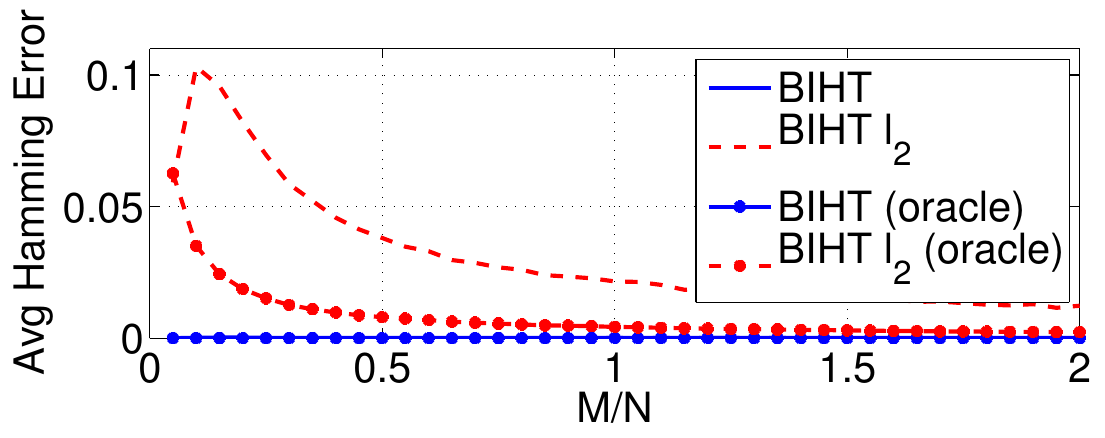}}
   \subfigure[]{\includegraphics[width=.9\imgwidth]{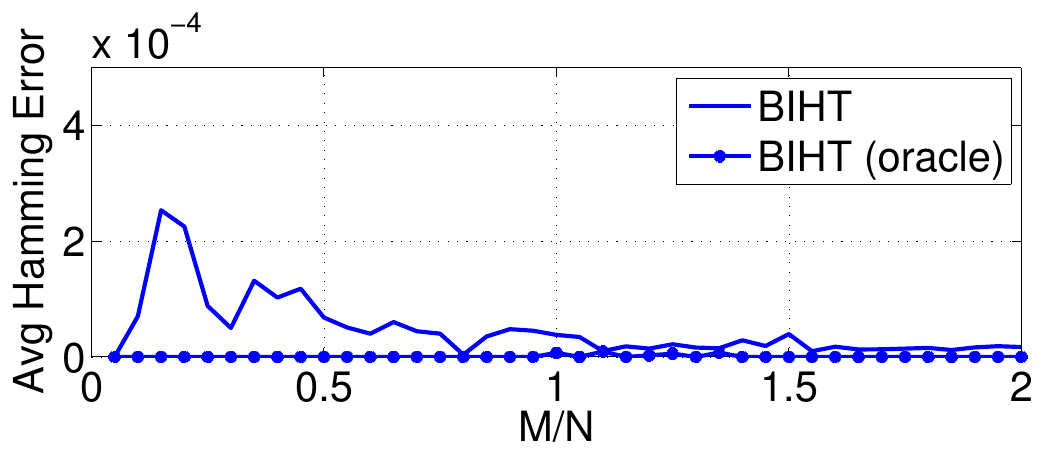}}
   \caption{Enforcing consistency: One-sided $\ell_{1}$ vs. one-sided
     $\ell_{2}$ BIHT.  When BIHT attempts to minimize a one-sided
     $\ell_{2}$ instead of a one-sided $\ell_{1}$ objective, the
     performance significantly decreases.  We find this to be the case
     even when an oracle provides the true signal support a
     priori. Note: (c) is simply a zoomed version (b).}
   \label{fig:enforce}
\end{figure}

\textbf{One-sided $\ell_{1}$ vs.\ one-sided $\ell_{2}$ objectives.}
As demonstrated in Figures~\ref{fig:avgerror} and \ref{fig:AngVsHamm},
the BIHT algorithm achieves significantly improved performance over
MSP and RSS in both angular error and Hamming error (consistency).  A
significant difference between these algorithms and BIHT is that MSP
and RSS seek to impose consistency via a one-sided $\ell_{2}$-norm, as
described in Section~\ref{subsec:bihtvar}.

Minimizing either the one-sided $\ell_{1}$ or one-sided $\ell_{2}$
objectives will enforce consistency on the measurements of the
solution; however, the behavior of these two terms appears to be
significantly different, according to the previously discussed
experiments.

To test the hypothesis that this term is the key differentiator
between the algorithms, we implemented BIHT-$\ell_{2}$, a one-sided $\ell_{2}$
variation of the BIHT algorithm that enabled a fair comparison of the one-sided objectives (see Section~\ref{subsec:bihtvar} for details).  We compared both the angular error and Hamming error
performance of BIHT and BIHT-$\ell_{2}$.  Furthermore, we implemented
\emph{oracle assisted} variations of these algorithms where the true
support of the signal is given a priori, \ie $\eta_K$ in (\ref{alg:biht2}) is
replaced by an operator that always selects the true support, and
thus the algorithm only needs to estimate the correct
coefficient values.  The oracle assisted case can be thought of
as a ``best performance'' bound for these algorithms.  Using these algorithms, we perform the same experiment detailed at the beginning of the section.

The results are depicted in
Figure~\ref{fig:enforce}.  The angular error behavior of
BIHT-$\ell_{2}$ is very similar to that of MSP and RSS and
underperforms when compared to BIHT.  We see the same situation with
regards to Hamming error: BIHT finds consistent solutions for
the majority of trials, but
BIHT-$\ell_{2}$ does not.

Thus, the results of this simulation suggest that the one-sided term
plays a significant role in the quality of the solution obtained.

One way to explain the performance discrepancy between the two
objectives comes from observing the connection between our
reconstruction problem and binary classification.  As explained
previously, in the classification context, the one-sided $\ell_{1}$
objective is similar to the hinge-loss, and furthermore, the one-sided
$\ell_{2}$ objective is similar to the so-called
\emph{square-loss}. Previous results in machine learning have shown
that for typical convex loss functions, the minimizer of the hinge
loss has the tightest bound between expected risk and the Bayes
optimal solution~\cite{RosDe-Cap::2004::Are-loss-functions} and good
error rates, especially when considering robustness to
outliers~\cite{RosDe-Cap::2004::Are-loss-functions,SreSriTew::2010::Smoothness-low-noise}.
Thus, the hinge loss is often considered superior to the square loss
for binary classification.\footnote{Additional ``well-behaved'' loss
  functions (\eg the Huber-ized hinge loss) have been
  proposed~\cite{tibshirani1996regression} and a host of
  classification algorithms related to this problem
  exist~\cite{BarFreLee::1998::Boosting-the-margin:,RatWar::2005::Efficient-margin,
    Blu::1998::On-Line-Algorithms,
    SreSriTew::2010::Smoothness-low-noise,tibshirani1996regression},
  both of which may prove useful in the $1$-bit CS framework in the
  future.} One might suspect that since the one-sided
$\ell_{1}$-objective is very similar to the hinge loss, it too should
outperform other objectives in our context. Understanding why in our
context, the geometry of the $\ell_{1}$ and $\ell_{2}$ objectives
results in different performance is an interesting open problem.

\begin{figure}
   \centering
   \subfigure[]{\includegraphics[width=.48\textwidth]{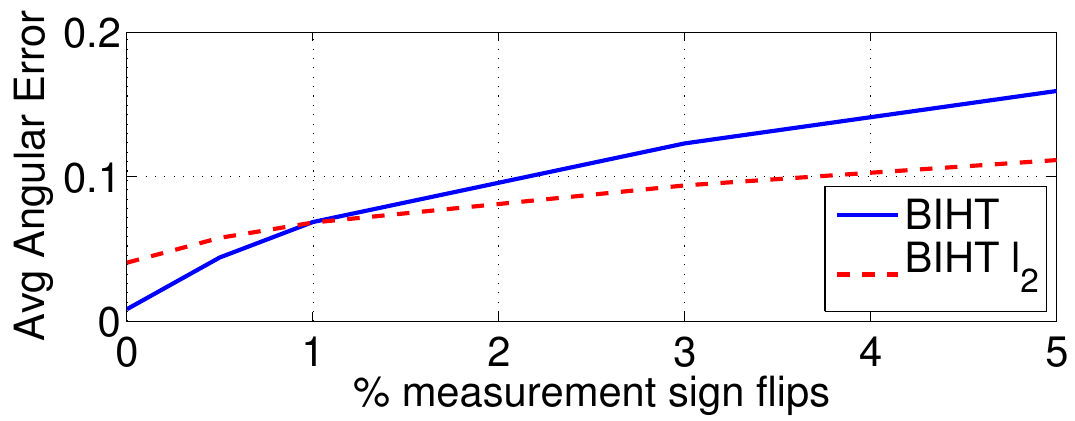}}
   \subfigure[]{\includegraphics[width=.48\textwidth]{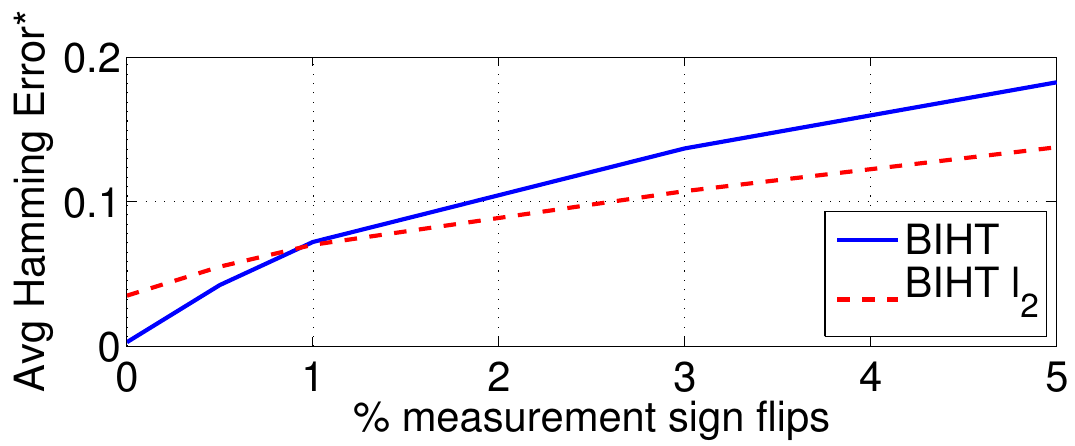}}
   \caption{Enforcing consistency with noise: One-sided $\ell_{1}$
     vs. one-sided $\ell_{2}$ BIHT.  When BIHT attempts to minimize a
     one-sided $\ell_{2}$ instead of the one-sided $\ell_{1}$
     objective, the algorithm is more robust to flips of measurement
     signs. *Note that the Hamming error in (b) is measured with
     regards to the \emph{noisy} measurements, \eg a Hamming error of
     zero means that we reconstructed the signs of the noisy
     measurements exactly.}
   \label{fig:enforcenoise}
\end{figure}

We probed the one-sided $\ell_1$/$\ell_2$ objectives further by
testing the two versions of BIHT on noisy measurements.  We flipped a number of measurement signs at random in each
trial.  For this experiment, $N=M=1000$ and $K=10$ are fixed, and we
performed 100 trials.  We varied the number of sign flips between
$0\%$ and $5\%$ of the measurements.  The results of the experiment
are depicted in Figure~\ref{fig:enforcenoise}.  We see that for both
the angular error in Figure~\ref{fig:enforcenoise}(a) and Hamming
error in Figure~\ref{fig:enforcenoise}(b), that the one-sided
$\ell_{1}$ objective performs better when there are only a few errors
and the one-sided $\ell_{2}$ objective performs better when there are significantly more errors.  This is expected since the $\ell_{1}$
objective promotes sparse errors. This experiment implies that 
BIHT-$\ell_{2}$ (and the other one-sided $\ell_{2}$-based algorithms) may be more useful when the measurements contain
significant noise that might cause a large number of sign flips, such
as Gaussian noise.

\begin{figure}[t]
   \centering
   \includegraphics[width=1\imgwidth]{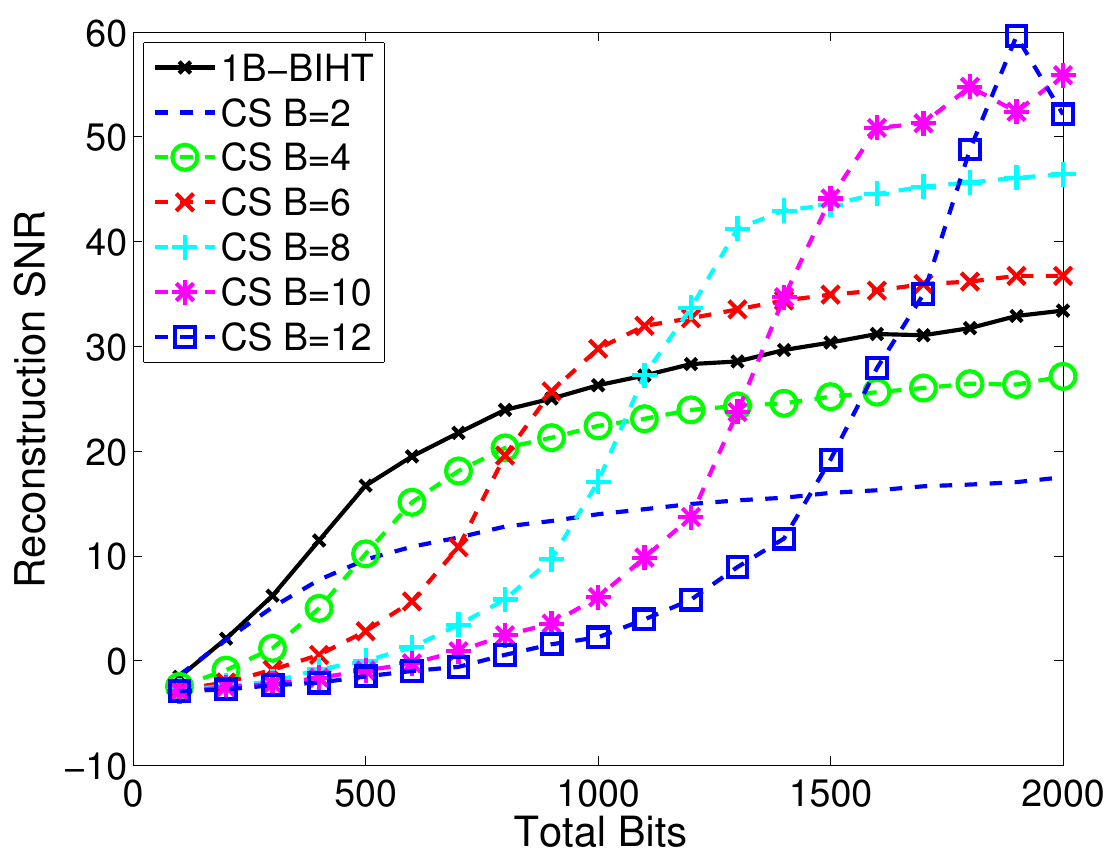} 
   \caption{Comparison of BIHT to conventional CS multibit uniform scalar quantization
     (multibit reconstructions performed using
     BPDN~\cite{CanRomTao::2006::Stable-signal}).  BIHT is competitive with standard
     CS working with multibit measurements when the total number of bits is severely constrained.  In particular, the BIHT
     algorithm performs strictly better than CS with $4$ bits per measurement.}
   \label{fig:multibit}
\end{figure}

\textbf{Performance with a fixed bit-budget.}  In some applications we are interested in reducing the total number of bits acquired due to storage or communication costs.  Thus, given a fixed total number of bits, an interesting question is how well $1$-bit CS performs in comparison to conventional CS
quantization schemes and algorithms.  For the sake of brevity, we give
a simple comparison here between the $1$-bit techniques and uniform
quantization with \emph{Basis Pursuit DeNoising}
(BPDN)~\cite{CanRomTao::2006::Stable-signal} reconstruction.  While
BPDN is not the optimal reconstruction technique for quantized measurements, it
(and its variants such as the LASSO~\cite{tibshirani1996regression})
is considered a benchmark technique for reconstruction from
measurements with noise and furthermore, is widely used in practice.

The experiment proceeds as follows.  Given a total number of bits and
a (uniform) quantization bit-depth $B$ ({i.e.}, number of bits
per measurement), we choose the number of measurements as $M =
\mathrm{total~bits}/B$, $N=2000$, and the sparsity $K=20$.  The remainder of the experiment proceeds as
described earlier (in terms of drawing matrices and signals).   For bit depth greater than $1$, we reconstruct
using BPDN with an optimal choice of noise parameter and we scale the quantizer to such that signal can take full advantage of its dynamic range.  

The results of this experiment are depicted in
Figure~\ref{fig:multibit}.  We see a common trend in each line:
lackluster performance until ``sufficient'' measurements are acquired,
then a slow but steady increase in performance as additional
measurement are added, until a performance plateau is reached.  Thus,
since lower bit-depth implies that a larger number of measurements will be used, $1$-bit CS reaches the performance plateau earlier than in
the multi-bit case (indeed, the transition point is achieved at a higher number of total bits as the bit-depth is increased).  This enables significantly improved
performance when the rate is severely constrained and higher bit-rates
per measurements would significantly reduce the number of available
measurements. For higher bit-rates, as expected from the analysis
in~\cite{BouBar::2007::Quantization-of-sparse}, using fewer
measurements with refined quantization achieves better performance. 

It is also important to note that, regardless of trend, the BIHT
algorithm performs strictly better than BPDN with $4$ bits per
measurement and uniform quantization for the parameters tested
here. This gain is consistent with similar gains observed
in~\cite{BouBar::2008::1-Bit-compressive,Bou::2009::Greedy-sparse}.  A
more thorough comparison of additional CS quantization techniques with
$1$-bit CS is a subject for future study.

\begin{figure}[t]
   \centering
   \includegraphics[width=1\imgwidth]{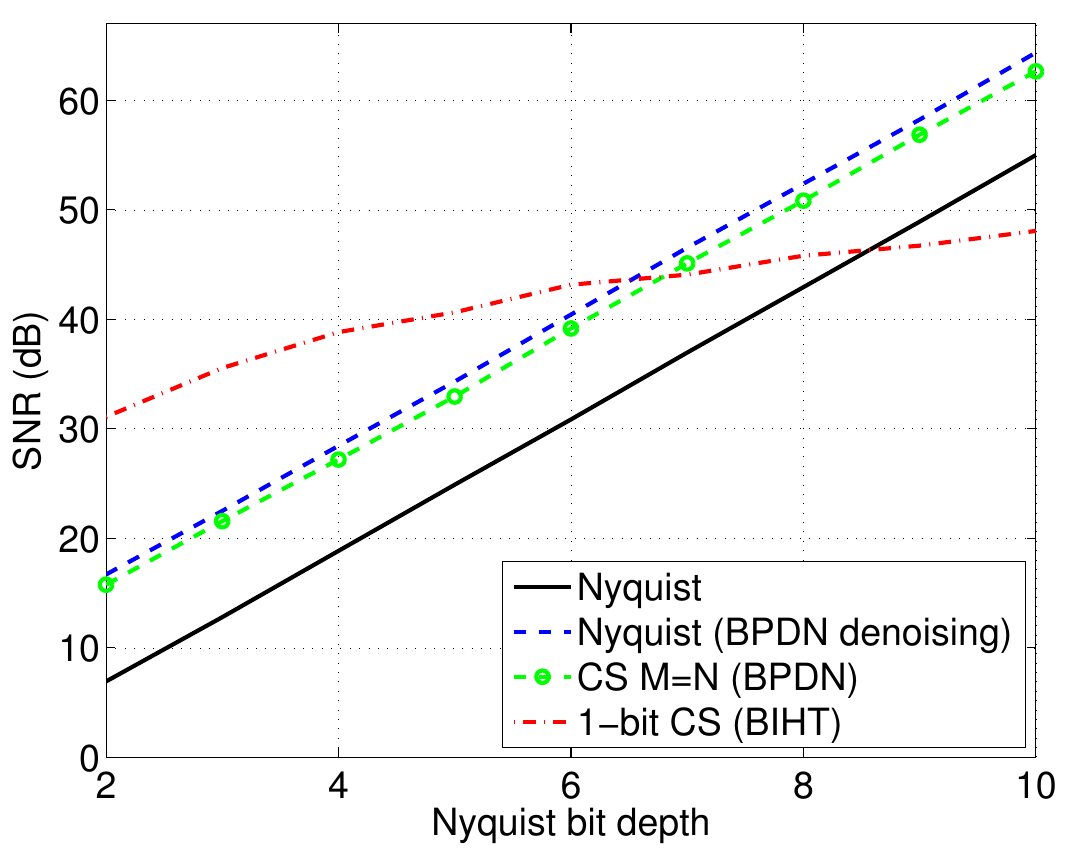} 
   \caption{Comparison of uniformly quantized Nyquist-rate samples
     with linear reconstruction (solid) and BPDN denoising
     (dashed), CS with $M=N$ and BPDN reconstruction (dash-circle), and $1$-bit quantized CS measurements with BIHT
     reconstruction (dash-dotted). Nyquist samples were quantized with
     bit-depth $\beta \in [2,10]$ and $1$-bit CS used $M=\beta N$
     measurements; the same number of bits is used in each
     reconstruction.  The Nyquist-rate lines have the classical $6.02$dB/bit-depth
     slope, as expected.  For a fixed number of bits, $1$-bit CS does
     \emph{not} follow this slope and outperforms conventional quantization
     when $\beta < 6$.}
   \label{fig:comparisonNyquist}
\end{figure}
\textbf{Comparison to quantized Nyquist samples.}  In our final
experiment, we compare the performance of the $1$-bit CS technique to
the performance of a conventional uniform quantizer applied to uniform
Nyquist-rate samples.  Specifically, in each trial we draw a new
Nyquist-sampled signal in the same way as in our previous experiments
and with fixed $N=2000$ and $K=20$; however, now the signals are
sparse in the discrete cosine transform (DCT) domain.  We consider
four reconstruction experiments.  First, we quantize the Nyquist-rate
signal with a bit-depth of $\beta$ bits per time-domain sample (and optimal
quantizer scale) and perform linear reconstruction (i.e., we just use
the quantized samples as sample values).  Second, we apply BPDN to the
quantized Nyquist-rate samples with optimal choice of noise parameter,
thus denoising the signal using a sparsity model. Third, we draw a new
Gaussian matrix with $M=N$, quantize the measurements to $\beta$ bits,
again at optimal quantizer scale, and reconstruct using BPDN. Fourth,
we draw a new Gaussian matrix with $M = \beta N$ and compute
measurements, quantize to one bit per measurement by maintaining their
sign, and perform reconstruction with BIHT.  Note that the same total
number of bits is used in each experiment.

Figure~\ref{fig:comparisonNyquist} depicts the average SNR obtained by
performing $100$ of the above trials.  The linear, BPDN, Gaussian
measurements with BPDN, and BIHT reconstructions are depicted by
solid, dashed, dash-circled, and dash-dotted lines, respectively.  The
linear reconstruction has a slope of $6.02$dB/bit-depth, exhibiting a
well-known trade-off for conventional uniform quantization.  The BPDN
reconstruction (without projections) follows the same trend, but
obtains an SNR that is at least $10$dB higher than the linear
reconstruction.  This is because BPDN imposes the sparse signal model
to denoise the signal.  We see about the same performance with the
Gaussian projections at $M=N$, although it performs slightly worse
than without projections since the Gaussian measurements require a
slightly larger quantizer range. Similarly to the results in
Fig.~\ref{fig:multibit}, in low Nyquist bit-depth regimes ($\beta <
6$), $1$-bit CS achieves a significantly higher SNR than the other two
techniques.  When $6 < \beta < 8$, $1$-bit
CS is competitive with the BPDN scenario. Thus, for a fixed number of bits, $1$-bit
CS is competitive to conventional sampling with uniform quantization,
especially in low bit-depth regimes.  

\section{Discussion}
\label{sec:disc}

In this paper we have developed a rigorous mathematical foundation for
$1$-bit CS.  Specifically, we have demonstrated a lower bound on
reconstruction error as a function of the number of measurements and
the sparsity of the signal. We have demonstrated that Gaussian random
projections almost reach this lower bound (up to a log factor) in the
noiseless case.  This behavior is consistent with and extends existing
results in the literature on multibit scalar quantization and 1-bit
quantization of non-sparse signals.

We have also introduced reconstruction robustness guarantees through
the binary $\epsilon$-stable embedding (B$\epsilon$SE) property. This
property can be thought of as extending the RIP to $1$-bit quantized
measurements.  To our knowledge, this is the first time such a
property has been introduced in the context of quantization. To be
able to use this property we showed that random constructions using
Gaussian pdf (or more generally using isotropic pdf in the signal
space $\Rbb^N$) generate such embeddings with high probability. This
construction class is still very limited compared to the numerous
random constructions known for generating RIP matrices. Extending this
class with other constructions is an interesting topic for future
research.

Using the B$\epsilon$SE, we have proven that $1$-bit CS systems are robust to measurement noise added before
quantization as well as to signals that are not exactly sparse but
compressible.

We have introduced a new $1$-bit CS algorithm,
BIHT, that achieves better performance over previous algorithms in the
noiseless case. This improvement is due to the enforcement of
consistency using a one-sided linear objective, as opposed to a quadratic one. The
linear objective is similar to the hinge loss from
the machine learning literature. 

We remind the reader that the central goal of this paper has been
signal \emph{acquisition} with quantization. As explained previously,
one motivation for our work is the development of very high speed
samplers.  In this case, we are interested in building fast samplers by relaxing the requirements on the primary hardware burden, the quantizer. Such devices are susceptible to noise.  Thus, while our noiseless results extend previous $1$-bit
quantization results (\eg see \cite{AndDatImm::2006::Locality-sensitive-hashing}
and \cite{vivekQuantFrame}) to the sparse signal model setting and are
of theoretical interest, a major contribution has been the further
development of the robust guarantees, even if they produce error rates that seem suboptimal when compared to the noiseless case.  

A number of interesting questions remain unanswered. As we discuss in
Section~\ref{subsec:results} earlier, we have found that the
B$\epsilon$SE holds for Gaussian matrices with angular error decay
roughly on the order of $O(\sqrt{K/M})$ worse than optimal. One
question is whether this gap can be closed with an alternative
derivation, or whether it is a fundamental requirement for stability.
Another useful pursuit would be to provide a more rigorous
understanding of the discrepancy between the performance of the
one-sided $\ell_{1}$ and $\ell_{2}$ objectives. Analysis of the
performance behavior might lead to better one-sided functions.

\section*{Acknowledgments}

Thanks to Rachel Ward for pointing us to the right reference with
regards to the lower bound (\ref{eq:tight}) used in
Appendix~\ref{app:subspace_intersections} and recommending several
useful articles as well as Vivek Goyal for pointing us to additional
prior work in this area.  Thanks to Zaiwen Wen and Wotao Yin for
sharing some of the data
from~\cite{LasWenYin::2010::Trust-but-verify:} for our algorithm
comparisons, as well as engaging in numerous conversations on this
topic.  Thanks also to Nathan Srebro for his advice and discussion
related to the one-sided penalty comparison and connections to binary
classification, and thanks to Amirafshar Moshtaghpour for his
correction of a small error in the bounds of Appendix~\ref{app:proof_epsilon_MK_complex}. Finally, thanks to Yaniv Plan and anomynous reviewers
for their useful advices and remarks for improving this paper, and in
particular the proof of Theorem~1 with respect to the optimality of
the $\Sigma^*_K$ covering.

\appendix

\section{Lemma~\ref{res:intersecions}: Intersections of Orthants by Subspaces}
\label{app:subspace_intersections}
While there are $2^{M}$ available quantization points provided by
$1$-bit measurements, a $K$-sparse signal will not use all of them.
To understand how effectively the quantization bits are used, we first
investigate how the $K$-dimensional subspaces projected from the
$N$-dimensional $K$-sparse signal spaces intersect orthants in the
$M$-dimensional measurement space, as shown in
Fig~\ref{fig:orthant_geometry} for $K=2$ and $M=3$.

\begin{figure}
  \centering
  \subfigure[\label{fig:simple_geom}]{
     \raisebox{5mm}{\includegraphics[height=3.5cm]{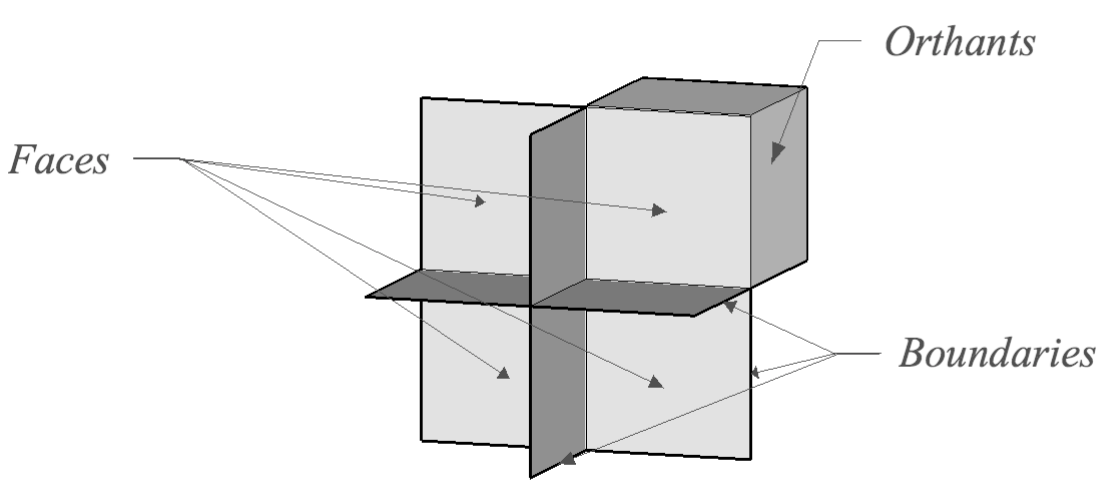}}
   }
   \subfigure[\label{fig:sph_cap}]
   {
    \includegraphics[height=4.3cm]{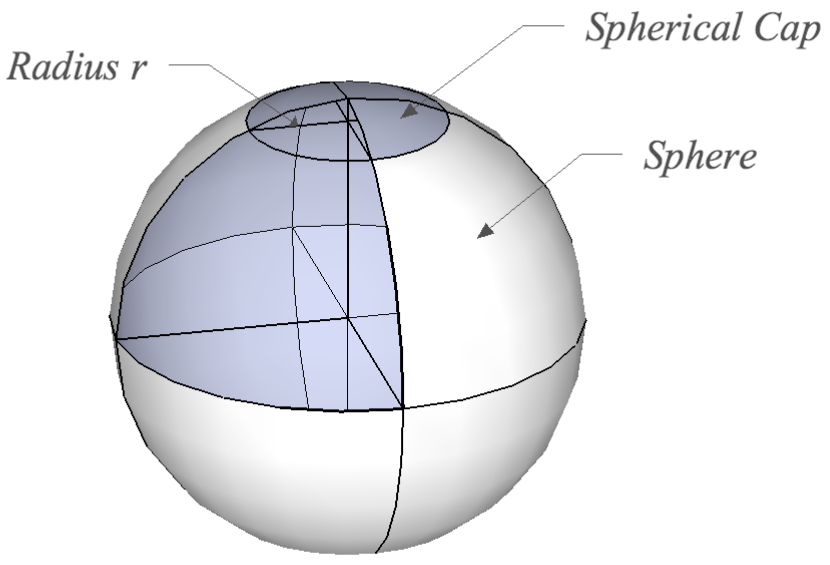}
  }
  \caption{(a) The geometry of orthants in $\mathbb{R}^{3}$. (b)
    The geometry of spherical caps.}
\end{figure}

We use $I(M,K)$ to denote the maximum number of orthants in $M$
dimensions intersected by a $K$-dimensional subspaces. A bound of for
$I(M,K)$ is developed in~\cite{coverSphereCaps,bib:Flatto70}:
\begin{align}
  I(M,K) \le 2\sum_{l=0}^{K-1}\binom{M-1}{l}.  \label{eq:tight}
\end{align}
For $K \leq M/2$, this simplifies to $I(M,K) \leq 2K {M-1
\choose K-1}$.

Using ${p \choose q-1}+{p \choose q} = {p+1 \choose q}$ we can also
derive a simple bound on (\ref{eq:tight}) for $K\le M$. We observe
that
$$
\sum_{l=0}^{K-1}\binom{M-1}{l} = \binom{M}{1} + \sum_{l=2}^{K-1}\binom{M-1}{l}
\leq \binom{M}{1} + \sum_{l=2}^{K-1}\binom{M}{l}
\leq \sum_{l=1}^{K-1}\binom{M}{l}.  
$$
Repeating the same argument, we find $\sum_{l=1}^{K-1}\binom{M}{l} \leq
\sum_{l=2}^{K-1}\binom{M+1}{l}\leq \cdots \leq \binom{M +
K - 2}{K - 1}$ and finally
\begin{align}
  I(M,K) \leq 2 \binom{M+K-2}{K-1} \leq 2 \binom{2M-1}{K-1} = 2
  (\tfrac{K}{2M})\binom{2M}{K} \leq  2^{K}\left(\frac{eM}{K}\right)^K,
  \label{eq:tight_bounded}
\end{align}
using the bound
$\binom{M}{K}\le\left(\tfrac{eM}{K}\right)^K$. 

While the bound in~\eqref{eq:tight} is tight and holds for subspaces
in a general configuration, the closed form simplified bound in~\eqref{eq:tight_bounded} can be improved by a factor of $(M-K+1)$
(which asymptotically makes no difference in the subsequent
development) using the proof we develop in the remainder of this
appendix. In addition to the improvement, the proof also provides
significant geometrical intuition to the problem.

First we define two new elements in the geometry of the problem:
orthant boundaries and their faces. Each orthant has $M$ boundaries
of dimension $M-1$, defined as the subspace with a coordinate set to
0:
\begin{align*}
  \mathfrak{B}_i=\{\bs x~|~(\bs x)_i=0\}.
\end{align*}
We split each boundary into $2^{M-1}$ faces, defined as the set
\begin{align*}
  \mathcal{F}_{i,\sv}=\left\{\bs x~|~(\bs x)_i=0~\mbox{and~}\sign{(\bs x)_j}=
  (\sv)_j\mbox{~for~all~}j\ne i\right\},
\end{align*}
where $\sv$ is the sign vector of a bordering orthant, and $i$ is
the boundary in which the face lies. Each face borders two
orthants. Note that the faces are $(M-1)$-dimensional orthants in the
$(M-1)$-dimensional boundary subspace. The geometry of the problem in
$\mathbb{R}^{3}$ is summarized in Figure~\ref{fig:simple_geom}.

Next, we upper bound $I(M,K)$ using an inductive argument that relies on
 the following two lemmas:
\begin{lemma}
  \label{res:subspace_boundary_intersec}
  If a $K$-dimensional subspace $\mathcal{S}\subset \Rbb^M$ is not the
  subset of a boundary $\mathfrak{B}_i$, then the subspace and
  boundary do intersect and their intersection is a $(K-1)$-dimensional
  subspace of $\mathfrak{B}_i$.
\end{lemma}
\begin{proof}
  We count the dimensions of the relevant spaces. If $\mathcal{S}$ is
  not a subset of $\mathfrak{B}_i$, then it equals the direct sum
  $\mathcal{S}=(\mathcal{S} \cap\mathfrak{B}_i)\oplus\mathcal{W}$,
  where $\mathcal{W}\subset \Rbb^M$ is also not a subspace of
  $\mathfrak{B}_i$. Since $\dim{\mathfrak{B}_i}=M-1$,
  $\dim{\mathcal{W}}\le 1$, and
  $\dim{\mathcal{S}\cap\mathfrak{B}_i}=K-1$ follows.
\end{proof}

\begin{lemma}
  \label{res:subspace_hyperq_intersec}
  For $K>1$, a $K$-dimensional subspace that intersects an orthant also
  non-trivially intersects at least $K$ faces bordering that
  orthant.
\end{lemma}
\begin{proof}
  Consider a $K$-subspace $\mathcal{S}$, a point $\bs p\in\mathcal{S}$
  interior to the orthant $\mathcal{O}_{\sign{\bs p}}$, and a vector
  $\bs x_1\in\mathcal{S}$ non-parallel to $\bs p$. The following
  iterative procedure can be used to prove the result:
  \begin{enumerate}
  \item Starting from 0, grow $a$ until the set $\bs p\pm a\bs x_l$
    intersects a boundary $\mathfrak{B}_i$, say at $a=a_l$. It is
    straightforward to show that as $a$ grows, a boundary will be
    intersected. The point of intersection is in the face
    $\mathcal{F}_{i,\sign{\bs p}}$. The set $\{\bs p\pm a\bs x_l|
    a\in(0,a_l)\}$ is in the orthant $\mathcal{O}_{\sign{\bs p}}$.
  \item Determine a vector $\bs x_{l+1}\in\mathcal{S}$ parallel to all the
    boundaries already intersected and not parallel to $\bs p$, set
    $l=l+1$ and iterate from step 1.
  \end{enumerate}
  A vector can always be found in step 2 for the first $K$ iterations
  since $\mathcal{S}$ is $K$-dimensional. The vector~is parallel to all the
  boundaries intersected in the previous iterations and therefore
  $\bs p\pm a\bs x_l$ always intersects a boundary not intersected
  before. Therefore, at least $K$ distinct faces are intersected.
\end{proof}

Lemmas~\ref{res:subspace_boundary_intersec}
and~\ref{res:subspace_hyperq_intersec} lead to the main result in this
Appendix. Lemma~\ref{res:intersecions} in
Section~\ref{subsec:opt} follows trivially.

\begin{lemma}
  The number of orthants intersected by a $K$-dimensional subspace  $\mathcal{S}$ in an $M$-dimensional space $\cl V$ is upper bounded by
  \begin{align*}
    I(M,K)\le \frac{2^K}{M-K+1}\binom{M}{K}\le2^K\binom{M}{K}.
  \end{align*}
\end{lemma}
\begin{proof}
The main intuition is that since the faces on each boundary are
equivalent to orthants in the lower dimensional subspace of the
boundary, the maximum number of faces intersected at each boundary is
a problem of dimension $I(M-1,K-1)$.
  
If $\mathcal{S}$ is contained in one of the boundaries in $\cl V$, the
number of orthants of $\cl V$ intersected is at most $I(M-1,K)$. Since
$I(M,K)$ is non-decreasing in $M$ and $K$, we can ignore this case in
determining the upper bound.

If $\mathcal{S}$ is not contained in one of the boundaries then
Lemma~\ref{res:subspace_boundary_intersec} shows that the
intersection of $\mathcal{S}$ with any boundary $\mathfrak{B}_i$ is a $(K-1)$-dimensional subspace in $\mathfrak{B}_i$. To count the faces of
$\mathfrak{B}_i$ intersected by $\mathcal{S}$ we use the observation in the
definition of faces above, that each face is also an orthant of
$\mathfrak{B}_i$. Therefore, the maximum number of faces of
$\mathfrak{B}_i$ intersected is a recursion of the same problem in
lower dimensions, \ie is upper bounded by $I(M-1,K-1)$. Since there
are $M$ boundaries in $\cl V$, it follows that the number of faces in
$\cl V$ intersected by $\mathcal{S}$ is upper bounded by $M\cdot I(M-1,K-1)$.

Using Lemma~\ref{res:subspace_hyperq_intersec} we know that for an
orthant to be intersected, at least $K$ faces adjacent to it should be
intersected. Since each face is adjacent to two orthants, the total
number of orthants intersected cannot be greater than twice the number
of faces intersected divided by $K$:
\begin{align}
  \label{eq:loose}
  I(M,K)\le\frac{2M\cdot I(M-1,K-1)}{K}.
\end{align}
To complete the induction we use $I(M,1)=2$ for all $M$; for $K=1$ the
subspace is a line through the origin, which can intersect only two
orthants\footnote{We recall that, from the definition (\ref{eq:orthant}), two
  different orthants have an empty intersection.}. This leads to:
$$
I(M,K) \leq \frac{2^K}{M-K+1} {M \choose K}\le2^K {M \choose K}\le 2^K\left(\frac{Me}{K}\right)^K.
$$
\end{proof}

\section{Theorem~\ref{res:lower_bound}: Distributing Signals
to Quantization Points}
\label{app:sph_cap}

To prove Theorem \ref{res:lower_bound} we consider how the available
quantization points optimally cover the set of signals $\Sigma^*_K=
\{\bs x \in\Rbb^N: \|\bs x\|=1, \|\bs x\|_0\leq K\}$. This set
corresponds to the union of $L={N \choose K}$ $K$-dimensional unit
spheres $\bigcup_{i\in [L]} S_i$, each $S_i = \{\bs x\in \Rbb^N: \|\bs
x\|=1,\, \supp \bs x \subset T_i\}$ being associated to one support
$T_i \subset [N]$ taken amongst the $L$ available $K$-length supports
of $\Rbb^N$. The cover should be optimal with respect to the worst
case distance, denoted by $r$, of any point in $\Sigma^*_K$ to its
closest quantization point. Our goal is to determine a lower bound on
the best-case $r$ we can achieve.

Unfortunately, determining the optimal cover of $\Sigma^*_K$ is not
straightforward. For instance, optimally covering each $S_i$
individually does not produce an optimal cover for their
union. Indeed, at the intersection of different $K$-dimensional
spheres in $\Sigma^*_K$ there are $K$-sparse signals, with different
support, very close to each other. A single quantization point in
$\Rbb^N$ could be close to all those signals, whereas independent
cover of each $S_i$ would have to use a different quantization point
to represent the signals on each sphere in that intersection. 

Thus, instead of determining the optimal cover of $\Sigma^*_K$, we
establish a lower bound on $r$ required to cover a subset $\widetilde
\Sigma^*_K$ of $\Sigma^*_K$ using the same number of points. A cover
of $\Sigma^*_K$ with a smaller $r$ would not be possible, since that
would also cover $\widetilde \Sigma^*_K$ with the same or smaller
$r$. Therefore, this $r$ establishes a lower bound for the cover of
$\Sigma^*_K$. To establish the bound, we pick $\widetilde \Sigma^*_K$
such that the neighborhood around the intersection of the balls is not
included. Specifically, we pick
$$
\widetilde \Sigma^*_K\ :=\ \bigcup_{i}\, \widetilde S_i \ \subset\
\Sigma^*_K,\quad \widetilde S_i := \{\bs x \in S_i: \forall\, k\in
T_i,\, |x_k| > 2r\} \subset S_i.  
$$ In other words, we pick the subset of each $K$-ball, such that all
the nonzero coordinates of the signals in the subset are greater than
$2r$. The union of those subsets for all possible supports comprises
$\widetilde \Sigma^*_K$. This choice ensures that any signal $\bs
x\in\widetilde S_i$ has distance at least $2r$ from any signal in any
other $\widetilde S_j, j\ne i$, and therefore both cannot be close to
a common quantization point in $\Rbb^N$ with distance $r$ from
each. Notice that $\widetilde S_i$ is non empty as soon as
$r<1/(2\sqrt K)$.

With this choice of $\widetilde\Sigma_K^*$, an optimal covering can be
obtained by merging the optimal coverings of each $\widetilde S_i$ for
$i\in [L]$. For an optimal covering of distance $r$, each of the
elements in $\widetilde S_i$ should belong to some ball of radius $r$
centered at the quantization point. Thus, each quantization point and
its corresponding $r$-ball should cover as large an area of
$\widetilde S_i$ as possible. This is achieved when the quantization
point is on $\widetilde S_i$, and the intersection of the ball with
$\widetilde S_i$ is a spherical cap of radius
$r$~\cite{bib:Ball97}. Thus, for the spherical caps to cover
$\widetilde S_i$, the total area of all spherical caps in the cover
should be greater than the area of $\widetilde S_i$. Furthermore,
since the overall cover of $\widetilde \Sigma^*_K$ is composed of
separate cover of each $\widetilde S_i$, the total area of all
spherical caps used for the cover should be greater than the area of
$\widetilde \Sigma^*_K$. 

Therefore, since $2^KL\binom{M}{K}$ quantization points and
corresponding spherical caps are available, according to
Lemma~\ref{res:intersecions} and
Corollary~\ref{cor:quantiz-sparsity-subspace}, to cover $L$ subsets of
$K$-spheres, as described above, the cover should satisfy
\begin{align}
  2^KL\binom{M}{K} \kappa(r) \geq \sigma(\widetilde
  \Sigma^*_K),\label{eq:area_ineq}
\end{align}
where $\sigma(\cdot)$ denotes the rotationally invariant area
measure the $K$-sphere $S_i$ and $\kappa(r)$ denotes the surface of a
spherical cap of radius $r$.

To determine the smallest $r$ satisfying~\eqref{eq:area_ineq}, we thus
need to measure the set $\widetilde \Sigma^*_K$. Choosing one $i\in
[L]$ we have $\sigma(\widetilde \Sigma^*_K) = L \sigma(\widetilde
S_i)$ since the sets $\widetilde S_k$ ($k\in [L]$) are disjoint with
identical area. We first show that
\begin{align}
  \sigma(\widetilde S_i) \geq (1 - 2r \sqrt K)^K \sigma(S^{K-1}),\label{eq:cap_ineq}
\end{align}
where $S^{K-1}$ is a $K$-dimensional sphere in $\Rbb^K$. Note that
this bound is tight at two ends: at $r=0$ where $\widetilde S_i = S_i$
(almost everywhere) and at $r=1/(2\sqrt K)$ for which $\widetilde S_i
= \emptyset$.

To prove~\eqref{eq:cap_ineq} we assume without loss of generality $T_i
= \{1,\,\cdots, K\}$ and consider $S_i = S^{K-1} = \{\bs x \in \Rbb^K:
\|\bs x\| = 1\}$ and $\widetilde S_i = \{\bs x \in \Rbb^K: \|\bs x\| =
1, \,\min |x_k| > 2r\}$. We define the intersection $\widetilde S^+_i
:= \widetilde S_i \cap \cl P$ with the positive orthant $\cl P=\{\bs x
\in \Rbb^K: \forall\,i\in [K], x_i \geq 0\}$. By symmetry,
$\sigma(\widetilde S_i) = 2^K \sigma(\widetilde S^+_i)$ since there
are $2^K$ orthants in $\Rbb^K$. As described in~\cite{bib:Ball97}, the
ratio between the measure of $\widetilde S^+_i$ and the one of the
full sphere $S^{K-1}$ equals the ratio between the volume occupied by
a cone formed by $\widetilde S^+_i$ in the unit ball $B^K\subset
\Rbb^K$ and the volume of this ball, \ie
\begin{equation}
  \label{eq:area-to-vol-ratio}
\frac{\sigma(\widetilde S^+_i)}{\sigma(S^{K-1})} = \frac{\mu(\cl C(\widetilde
S^+_i))}{\mu(B^K)},  
\end{equation}
where $\mu$ is the Lebesgue measure in $\Rbb^K$ and $\cl C(A) = \{t\bs
a: t\in [0,1],\, \bs a\in A\} \subset B^K$ is the portion of the cone
(with apex on the origin and restricted to $B^K$) formed by the subset
$A\subset S^{K-1}$. The geometry of our problem is made clear in
Figure~\ref{fig:cover-explan}. 

\begin{figure}[!h]
  \centering
  \includegraphics[width=11cm]{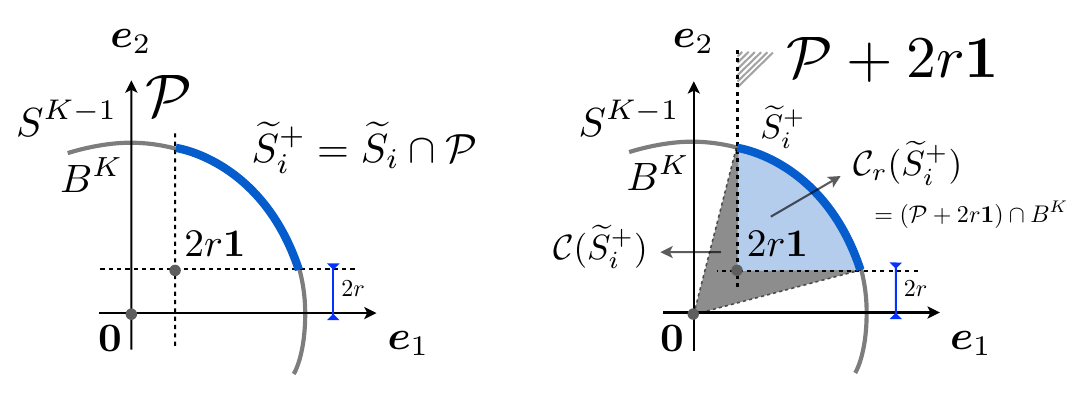}
  \caption{ The geometry of our problem in $\Rbb^2$. We only need to
    consider one orthant, $\cl P$, as the problem is the same for all
    orthants. To measure the surface area of $\widetilde S^+_i$
    relative to the surface of $S^{K-1}$, we consider the ratio of the
    volume of $\cl C(\widetilde S^+_i)$ with respect to this of $B^K$. We
    lower bound the former using $\cl C_r(\widetilde S^+_i) \subset
    \cl C(\widetilde S^+_i)$.}  \label{fig:cover-explan}
\end{figure}

To measure the volume of $\cl C(\widetilde S^+_i)$, writing $\bs
1=(1,\,\cdots,1)^T\in\Rbb^K$, we first define the set
$$
\cl C_r(\widetilde S^+_i) := (\cl P + 2r\bs 1)\cap B^K,
$$ also shown in the figure, which is non-empty for $2r \leq 1/\sqrt
K$. Since $\cl C_r(\widetilde S^+_i) \cap S_i = \widetilde S_i^+$ and
$2r\bs 1 \in \cl C(\widetilde S^+_i)$, it is straightforward to show
that $\cl C_r(\widetilde S^+_i)\subset \cl C(\widetilde S^+_i)$, and
therefore the measure of $\cl C_r(\widetilde S^+_i)$ is a lower bound
to the measure of $C(\widetilde S^+_i)$.
Furthermore, by the translation invariance of $\mu$,
$$
\mu(\cl
C_r(\widetilde S^+_i)) = \mu( (\cl P + 2r\bs 1)\cap B^K) = \mu(\cl
P\cap (B^K - 2r\bs 1)),
$$ where $B^K - 2r\bs 1$ is the unit ball centered on $-2r \bs 1$.
Setting $\alpha = (1 - 2r\sqrt K)$ it follows that $\|\bs u + 2r \bs
1\| \leq 1$ for any $\bs u\in\Rbb^K$, $\|\bs u\|\leq
\alpha$. Consequently, $\alpha B^K \subset B^K - 2r\bs 1$ and $ \cl P
\cap \alpha B^K \subset \cl P\cap (B^K - 2r\bs 1) $. Therefore, the
measure of the positive orthant of a $K$-ball with radius $\alpha$
lower bounds the measure of $\cl C(\widetilde S^+_i)$.

Putting everything together we obtain
$$
\mu(\cl C(\widetilde S^+_i))\ \geq\ \mu(\cl P \cap \alpha B^K) =
\alpha^{K} 2^{-K} \mu(B^K),
$$ which, using \eqref{eq:area-to-vol-ratio}, implies that
${\sigma(\widetilde S^+_i)}/{\sigma(S^{K-1})} \geq \alpha^{-K} 2^{-K}$
and $\sigma(\widetilde \Sigma^*_K) = L\sigma(\widetilde S_i) \geq (1 -
2r \sqrt K)^K L\sigma(S^{K-1})$.

In turn,~\eqref{eq:area_ineq} becomes
$$
2^KL\binom{M}{K} \kappa(r) \geq (1 - 2r \sqrt K)^K L\sigma(S^{K-1}),
$$ where $\kappa(r) \leq r^K \sigma(S^{K-1})$~\cite{bib:Ball97}. From 
$\binom{M}{K} \leq (eM/K)^K$, the result follows:
$$
2^K r^K (\tfrac{eM}{K})^K \geq (1 - 2r \sqrt K)^K\ \Rightarrow\ r \geq
(\tfrac{2eM}{K} + 2\sqrt K)^{-1} = K/(2eM + 2K^{3/2}). 
$$

\section{Theorem~\ref{thm:hash}: Optimal Performance via Gaussian Projections}
\label{sec:proof-hash}

To prove Theorem~\ref{thm:hash}, we follow the procedure given in
\cite[Theorem 3.3]{Boufounos::2010::univer_rate_effic_scalar_quant}.
We begin by restricting our analysis to the support set $T\subset
[N]:=\{1,\,\cdots,N\}$ with $|T| \leq D \leq N$, and thus we consider
vectors that lie on the (sub) sphere $\Sigma^*(T)=\{\bs x: \supp \bs
x\subset T, \|\bs x\|_{2}=1\}\subset \Rbb^N$.  We remind the reader
that $B_{r}(\bs x) := \{ \bs a \in S^{N-1}: \|\bs x - \bs a\|_{2} < r
\}$ is the ball of unit norm vectors of Euclidean distance $r>0$
around $\bs x$, and we write $B_r^*(\bs x) = B_r(\bs
x)\cap\Sigma^*(T)$ as in Section~\ref{subsec:results}.

Let us fix a radius $\delta>0$ to be precised later. The sphere $\Sigma^*(T)$ can be
covered with a finite set $ Q_\delta\subset \Sigma^*(T)$ of no more
than $(3/\delta)^D$ points such that, for any $\bs w\in \Sigma^*(T)$,
there exists a $\bs q\in Q_\delta$ with $\bs w\in B^{*}_{\delta}(\bs
q)$ \cite{BarDavDeV::2008::A-Simple-Proof}.

Using the notation $d_S$ defined in Sec.~\ref{subsec:prelim}, given a
vector $\bs \varphi \sim \SGR{N}{1}$ and two distinct points $\bs p$
and $\bs q$ in $Q_\delta$, we have that
$$
\bb P\big[\forall\, \bs u\in B_\delta^*(\bs p), \forall \bs v\in B_\delta^*(\bs
q): \sign\bs\varphi^T\bs u \neq \sign\bs\varphi^T\bs v\big]\ \geq\ \ang{\bs p}{\bs q}\ -\
\sqrt{\tfrac{\pi}{2} D}\,\delta,
$$ 
from Lemma~\ref{lemma:p0-p1-lower-bounds} (given in
Appendix~\ref{sec:proof-lemma-conc-prop-hamming-balls}).  Since for
all $\bs u\in B_\delta^*(\bs p)$ and $\bs v\in B_\delta^*(\bs
q)$
$$ 
\pi\,\ang{\bs p}{\bs q} \geq 2\sin(\tfrac{\pi}{2}\,\ang{\bs p}{\bs q})
= \|\bs p -\bs q\|_{2} 
\geq \|\bs u -\bs v\|_{2} - 2\delta,
$$
we can write for any $\epsilon_{o} > 0$ 
$$
\bb P\big[\forall\, \bs u\in  B_\delta^*(\bs p), \forall \bs v\in  B_\delta^*(\bs
q): \sign\bs\varphi^T\bs u \neq \sign\bs\varphi^T\bs v\ |\ \|\bs u -\bs v\|_{2}
> \epsilon_{o}\big]\ \geq\ \tfrac{\epsilon_{o}}{\pi}\ -\
(\tfrac{2}{\pi} + \sqrt{\tfrac{\pi}{2} D})\,\delta.
$$
By setting $\delta = \pi\epsilon_{o}/(4 + \pi\sqrt{2\pi
D})$ (and reversing the inequality), we obtain
\begin{align*}
&\bb P\big[\exists\, \bs u\in  B_\delta^*(\bs p), \exists \bs v\in B_\delta^*(\bs
q): \sign(\bs\varphi^T\bs u) = \sign(\bs\varphi^T\bs v)\ |\ \|\bs u -\bs v\|_{2}
> \epsilon_{o}\big]\ \leq\ 1 - \tfrac{\epsilon_{o}}{2}.
\end{align*}
Thus, for $M$ different random vectors $\bs \varphi_i$ arranged in
$\Phi=(\bs \varphi_1,\,\cdots, \bs \varphi_M)^T \sim \SGR{M}{N}$, and
for the associated mapping $A$ defined in (\ref{eq:defh}), we
get
\begin{align*}
&\bb P\big[\exists\, \bs u\in  B_\delta^*(\bs p), \exists \bs v\in  B_\delta^*(\bs
q): A(\bs u) = A(\bs v)\ |\ \|\bs u -\bs v\|_{2}
> \epsilon_{o}\big] \leq (1 - \tfrac{\epsilon_{o}}{2})^M.
\end{align*}
In other words, we have found a bound on the probability that two vectors'
measurements are consistent, even if their Euclidean distance is
greater than $\epsilon_{o}$, but only for vectors in the restricted
(sub) sphere $\Sigma^{*}(T)$.  Now we seek to cover the rest of the
space $\Sigma_{K}^{*}$ (unit norm $K$-sparse signals).

Since there are no more than ${ |Q_\delta| \choose 2} \leq (|
Q_\delta|)^2 \leq (3/\delta)^{2D}$ pairs of distinct points in $
Q_\delta$, we find
$$
\bb P\big[\exists\, \bs u, \bs v \in \Sigma^*(T): A(\bs u) = A(\bs v)\ |\ \|\bs u -\bs v\|_{2}
> \epsilon_{o}\big]\ \leq\ \big(\tinv{\pi\epsilon_{o}}(12 + 3\pi\sqrt{2\pi D})\big)^{2D}\ (1 - \tfrac{\epsilon_{o}}{2})^M.
$$

To obtain the final bound, we observe that any pair of unit $K$-sparse vectors $\bs x$ and $\bs s$ in
$\Sigma_K^*$ belongs to some $\Sigma^*(T)$ with $T=\supp \bs x \cup \supp
\bs s$ and $|T| \leq 2K$. There are no more than ${N \choose
2K} \leq (eN/2K)^{2K}$ of such sets $T$, and thus setting $D=2K$ above yields
\begin{align*}
&\bb P\big[\exists\, \bs u, \bs v \in \Sigma^*_K: A(\bs u) = A(\bs v)\ |\ \|\bs u -\bs v\|_{2}
> \epsilon_{o}\big]\\
&\leq\ (\tfrac{eN}{2K})^{2K}\ (\tinv{\pi\epsilon_{o}}(12 + 6\pi\sqrt{\pi
K}))^{4K}\ (1 - \tfrac{\epsilon_{o}}{2})^M\\
&\leq \exp\big[2K \log(\tfrac{eN}{2K}) + 4K \log(\tinv{\pi\epsilon_{o}}(12 + 6\pi\sqrt{\pi
K})) - M\tfrac{\epsilon_{o}}{2} \big], 
\end{align*}
where the second inequality follows from $1-\frac{\epsilon_{o}}{2}\leq
\exp\frac{\epsilon_{o}}{2}$. By upper bounding this probability by
$\eta$ and solving for $M$, we obtain 
$$
M \geq \tfrac{2}{\epsilon_{o}}\,\big(2K\,\log\tfrac{eN}{2K} + 4K
\log(\tinv{\pi\epsilon_{o}}(12 + 6\pi\sqrt{\pi K})) + \log
\tinv{\eta}\big).
$$
Since $K\geq 1$, we have that $\inv{\pi}(12 + 6\pi\sqrt{\pi K}) < 17 \sqrt{\tfrac{2K}{e}}$,
and thus the previous relation is then satisfied when
\begin{align*}
M\ &\geq\ \tfrac{2}{\epsilon_{o}}\,\big(2K\,\log\tfrac{eN}{2K} + 4K
\log(\tfrac{17\sqrt{2}}{\sqrt e\,\epsilon_{o}}\sqrt{K}) + \log
\tinv{\eta}\big)\\
&=\ \tfrac{2}{\epsilon_{o}}\,\big(2K\,\log N + 4K
\log(\tfrac{17}{\epsilon_{o}}) + \log
\tinv{\eta}\big).
\end{align*}

\section{Lemma~\ref{lem:conc-prop-hamming-balls}: Concentration of Measure for $\delta$-Balls}
\label{sec:proof-lemma-conc-prop-hamming-balls}

Since $T\subset [N]$ is fixed with size $|T|=D$, proving Lemma
\ref{lem:conc-prop-hamming-balls} amounts to showing that, for any
fixed $\epsilon > 0$ and $0\leq\delta\leq1$, given a Gaussian matrix
$\Phi \in \mathbb{R}^{M\times D}$, the mapping $A: \mathbb{R}^{D}
\rightarrow \mathcal{B}^{M}$ defined as $A(\bs u)=\sign(\Phi\bs u)$,
and for any fixed $\bs x,\bs s\in S^{D-1}$, we have
\begin{equation*}
\bb P \left(\ \forall \bs u\in B^{*}_{\delta}(\bs x),\,\forall \bs v\in B^{*}_{\delta}(\bs s),\ \big|\,d_H\big(A(\bs u),A(\bs v)\big)\ -\ \ang{\bs x}{\bs
s}\,\big|\ \leq\ \epsilon + \sqrt{\tfrac{\pi}{2} D}\,\delta\, \right)\ \geq\ 1 -
2\,e^{-2\epsilon^2 M},
\end{equation*}
where, in this case, $B^*_\delta(\bs p) = (B_\delta(\bs p)\cap S^{D-1}) \subset \Rbb^D$ for any $\bs p\in\Rbb^D$.

Given $\bs u'\in B^{*}_{\delta}(\bs x)$ and $\bs v'\in B^{*}_{\delta}(\bs s)$,
the quantity $Md_H\big(A(\bs u'),A(\bs v')\big)$ is the sum $\sum_i
A_i(\bs u') \oplus A_i(\bs v')$, where $A_i(\bs u')$ stands for the
$i^{\rm th}$ component of $A(\bs u')$. For one index $1\leq i\leq M$ 
\begin{align*}
A_i(\bs u') \oplus A_i(\bs v')\ &\leq\ Z^+_i := \max\big\{ A_i(\bs u) \oplus A_i(\bs v) :\ \bs u\in
B^{*}_{\delta}(\bs x), \bs v\in
B^{*}_{\delta}(\bs s)\,\big\},\\
A_i(\bs u') \oplus A_i(\bs v')\ &\geq\ Z^-_i:= \min\big\{ A_i(\bs u) \oplus A_i(\bs v) :\ \bs u\in
B^{*}_{\delta}(\bs x), \bs v\in
B^{*}_{\delta}(\bs s)\,\big\},
\end{align*}
and therefore 
$$
Z^-\ := \sum_{i=1}^M Z^-_i \ \leq\ M\, d_H\big(A(\bs u'),A(\bs v')\big)\
\leq\ \sum_{i=1}^M Z^+_i\ =:\ Z^+.
$$

Of course, the occurrence of $Z_i^+ = 0$ ($Z_i^- = 1$) means that all
vector pairs taken separately in $B^{*}_{\delta}(\bs x)$ and $B^{*}_{\delta}(\bs s)$ have
consistent (or respectively, inconsistent) measurements on the $i^{\rm th}$
sensing component $A_i$. More precisely, since $\bs\varphi_i \sim \SGR{N}{1}$, $Z^\pm_i$ are binary random variables such that $\bb P[Z^+_i =
1]=1-p_0$ and $\bb P[Z^-_i = 1] = p_1$ independently of $i$, where the
probabilities $p_0$ and $p_1$ are defined by
\begin{align*}
p_0(\bs x,\bs s,\delta)&\,=\ \bb P[Z^+_i = 0]\ =\ \bb P\big[\,\forall \bs u\in B^{*}_{\delta}(\bs x), \forall\bs v\in
B^{*}_{\delta}(\bs s),\ A_i(\bs u) = A_i(\bs v)\,\big],\\ 
p_1(\bs x,\bs s,\delta)&\,=\ \bb P\big[\,\forall \bs u\in B^{*}_{\delta}(\bs x), \forall\bs v\in
B^{*}_{\delta}(\bs s),\ A_i(\bs u) \neq A_i(\bs v)\,\big].
\end{align*}

In summary, $Z^+$ and $Z^-$ are binomially distributed with $M$ trials and probability of success $1-p_0$ and $p_1$, respectively. 
Furthermore, we have that $\E Z^+ = M\,(1-p_0)$ and $\E
Z^- = M\,p_1$, thus by the Chernoff-Hoeffding inequality,
\begin{align*}
  &\bb P\big[\,Z^+\ >\ M\,(1-p_0) + M\epsilon\,\big]\ \leq\ e^{-2M\epsilon^2},\\ 
  &\bb P\big[\,Z^-\ <\ M\,p_1 - M\epsilon\,\big]\ \leq\ e^{-2M\epsilon^2}.
\end{align*}
This indicates that with a probability higher than $1-2e^{-2M\epsilon^2}$, we
have 
$$
p_1 - \epsilon \ \leq\ d_H\big(A(\bs u'), A(\bs v')\big)\
\leq\ (1-p_0) + \epsilon.
$$
The final result follows by lower bounding $p_0$ and $p_1$ as in Lemma~\ref{lemma:p0-p1-lower-bounds}.

\begin{lemma}
\label{lemma:p0-p1-lower-bounds}
Given $0\leq \delta < 1$ and two unit vectors $\bs x,\bs s\in S^{D-1}$, we have 
\begin{align}
\label{eq:p0-lower-bound}
p_0&= \bb P\big[\,\forall \bs u\in B^*_\delta(\bs x),\ \forall \bs v\in
B^*_\delta(\bs s),\,\sign\scp{\bs\varphi}{\bs u}\, =\, \sign \scp{\bs\varphi}{\bs v}\,\big]\ \geq\ 1\ - \ang{\bs x}{\bs s}\ -\
\sqrt{\tfrac{\pi}{2} D}\,\delta,\\
\label{eq:p1-lower-bound}
p_1&= \bb P\big[\,\forall \bs u\in B^*_\delta(\bs x),\ \forall \bs v\in
B^*_\delta(\bs s),\,\sign\scp{\bs\varphi}{\bs u}\, \neq\, \sign
\scp{\bs\varphi}{\bs v}\,\big]\ \geq\ \ang{\bs x}{\bs s}\ -\
\sqrt{\tfrac{\pi}{2} D}\,\delta.
\end{align}
\end{lemma}

\begin{proof}[Proof of Lemma \ref{lemma:p0-p1-lower-bounds}] 
  We begin by introducing some useful properties of Gaussian vector
  distribution. If $\bs \varphi\sim \SGR{D}{1}$, the probability that
  $\bs \varphi\in \cl A \subset \Rbb^D$ is simply the measure $\mu$ of
  $\cl A$ with respect to the standard Gaussian density $\gamma(\bs
  \varphi) = \tinv{(2\pi)^{D/2}}\,e^{-\|\bs \varphi\|^2/2}$, \ie
    \begin{equation*}
     \bb P[\,\bs \varphi\in \cl A\,] =
  \mu(\cl A) = \int_{\cl A}\ud^D \bs \varphi\ \gamma(\bs\varphi),
  \end{equation*}
  with $\mu(\Rbb^D)=1$. 
  It may be easier to perform this integration over a hyper-spherical
  set of coordinates measured in a basis defined by the vectors $\bs
  x$ and $\bs s$.  This is possible since the pdf $\gamma$ is
  rotationally invariant.

  Specifically, we consider the canonical basis $\cl E= \{\bs
  e_1,\,\cdots,\bs e_D\}$ of $\Rbb^D$ where, by using the cross
  product $\wedge$ in $\Rbb^D$, $\bs e_1:=(\bs x\,\wedge\,\bs s)\,/\,\|\bs
  x\,\wedge\,\bs s\|_{2}$, $\bs e_D:= \bs x$ and $\bs e_{D-1}:=\bs
  e_D\,\wedge\,\bs e_1$, while the other vectors $\{\bs e_k: 2\leq
  k\leq D-2\}$ are defined arbitrarily for completing the basis.  In
  this system, the ``$\bs x\bs s$'' plane is equivalent to the plane
  spanned by $\bs e_D$ and $\bs e_{D-1}$. Moreover, any vector $\bs
  \varphi \in \Rbb^D$ can be represented by the spherical coordinates
  $(r,\phi_1,\,\cdots,\phi_{D-1})$ where $r=\|\bs
  \varphi\|_2\in\Rbb_+$, $(\phi_1,\cdots,\phi_{D-2})\in [0,\pi]^{D-2}$
  corresponds to the vector angles in each dimension, and
  $\phi_{D-1}\in[0,2\pi]$ being the angle formed by the projection of
  $\bs \varphi$ in the ``$\bs x\bs s$'' plane with $\bs x=\bs e_D$.

  The change of coordinates between the Cartesian and the spherical
  representations of $\bs \varphi$ in $\cl E$ is then defined as
  $\varphi_1= r\cos\phi_1$, $\varphi_2 = r\sin\phi_1\cos\phi_2$, ...,
  $\varphi_{D-1} =
  r\,\sin\phi_1\,\cdots\,\sin\phi_{D-2}\cos\phi_{D-1}$, and
  $\varphi_{D} = r\,\sin\phi_1\,\cdots\,\sin\phi_{D-2}\sin\phi_{D-1}$,
  while, conversely, $r=\|\bs \varphi\|_{2}$, $\tan\phi_1 =
  (\varphi_D^2+\cdots+\varphi_2^2)^{1/2}/\varphi_1$, ...,
  $\tan\phi_{D-2} =
  (\varphi_D^2+\varphi_{D-1}^2)^{1/2}\,/\,\varphi_{D-2}$, and
  $\tan\phi_{D-1} =\varphi_D\,/\,\varphi_{D-1}$.\footnote{This change
    of coordinates can be very convenient. For instance, the proof of
    Lemma \ref{lem:conc-prop-hamming} relies on the computation $\bb
    P[A_i(\bs x) \neq A_i(\bs s)] = \mu(\cl A = \{\bs\varphi:
    \phi_{D-1}\in [0,\pi\,\ang{\bs x}{\bs
      s}]\cup[\pi,\pi+\pi\,\ang{\bs x}{\bs s}]\})=\ang{\bs x}{\bs s}$,
    since for (almost) all $\bs \varphi\in \cl A$, $\bs x$ and $\bs s$
    live in the two different subvolumes determined by the plane
    $\{\bs u:\scp{\bs\varphi}{\bs u}=0\}$
    \cite{GoeWil::1995::Improved-approximation,ShaJacSta::2010::Randomly-driven}.}

  We now seek a lower bound on $p_1$. Computing this probability
  amounts to estimating
$$
p_1 = \bb P[\,\forall \bs u\in B^*_\delta(\bs x),\ \forall \bs v\in
B^*_\delta(\bs s),\ \scp{\bs\varphi}{\bs u}\scp{\bs\varphi}{\bs v} \leq
0\,]\ =\ \mu(\cl W_{\delta}),
$$
where $\cl W_{\delta} := \{\bs\varphi: \scp{\bs\varphi}{\bs
  u}\scp{\bs\varphi}{\bs v} \leq 0,\ \forall \bs u\in B^*_\delta(\bs
x),\ \forall \bs v\in B^*_\delta(\bs s)\}$ is the set of all vectors
$\bs \varphi$ such that its inner product with $\bs u$ and $\bs v$
result in different signs. 

Note that if $B^*_\delta(\bs x)\cap B^*_\delta(\bs s) \neq \emptyset$,
then $p_1 = 0$ since $p_1\leq \bb P[\,\forall \bs u\in B^*_\delta(\bs
x)\cap B^*_\delta(\bs s), \ \scp{\bs\varphi}{\bs u}^2 = 0\,] = 0$.  This
non-empty intersection is avoided when $\ang{\bs x}{\bs s}\geq
\tfrac{4}{\pi}\arcsin\delta/2$. Furthermore, since $\arcsin\lambda\leq
\tfrac{\pi}{2} \lambda$ for any $0\leq \lambda\leq 1$, this occurs if
$\ang{\bs x}{\bs s}\geq \delta$.

The remainder of the proof is devoted to finding an appropriate way to integrate the set $\cl W_{\delta}$. To this end, we begin by demonstrating that estimating $p_1$ can be simplified with the following equivalence
(proved just after the completion of the proof of Lemma~\ref{lemma:p0-p1-lower-bounds}).

\begin{lemma}
\label{lemma-W-V-xy-equiv}
The set $\cl W_{\delta}\subset \Rbb^{D}$ is equal to the set
$$
\cl V^-_{\delta}\ =\ \{\bs \varphi: \scp{\bs\varphi}{\bs
  x}\scp{\bs\varphi}{\bs s} \leq 0, \|\bs x\,-\,\cl
P_{\Pi(\varphi)}\,\bs x\| \geq \delta, \|\bs s\,-\,\cl
P_{\Pi(\varphi)}\,\bs s\| \geq \delta\},
$$
where $\cl P_{\Pi(\bs \varphi)}$ is the orthogonal projection
on the plane $\Pi(\bs\varphi)=\{\bs u\in\Rbb^D: \scp{\bs \varphi}{\bs
  u}=0\}$.
\end{lemma}

Using the hyper spherical coordinate system developed earlier and
denoting the angle $\pi\,d_S(\bs x,\bs s)$ by $\theta$, membership in
$\cl V^-_{\delta}$ can be expressed as
\begin{empheq}[left={\bs \varphi =(r,\phi_1,\,\cdots,\phi_{D-1})\in
\cl V^-_{\delta}\ \Leftrightarrow\ \empheqlbrace}]{align}
&\tan \phi_{D-1}\ \in\ [0, \tan \theta],\tag{R1}\label{eq:R1}\\
&\sin\phi_1\,\cdots\,\sin\phi_{D-2}\,|\sin\phi_{D-1}|\ \geq\ \delta,\tag{R2}\label{eq:R2}\\
&\sin\phi_1\,\cdots\,\sin\phi_{D-2}\,|\sin(\phi_{D-1}-\theta)|\ \geq\ \delta.\tag{R3}\label{eq:R3}
\end{empheq}

Indeed, requirement \eqref{eq:R1} enforces $\scp{\bs \varphi}{\bs x}\scp{\bs
  \varphi}{\bs s}\leq 0$, while \eqref{eq:R2} and \eqref{eq:R3} are direct translations
of the requirements that $\|\bs x\,-\,\cl P_{\Pi(\varphi)}\,\bs x\| =
|\scp{\widehat{\bs \varphi}}{\bs x=\bs e_D}| \geq \delta$ and $\|\bs
s\,-\,\cl P_{\Pi(\varphi)}\,\bs s\| = |\scp{\widehat{\bs \varphi}}{\bs
  s=-\sin\theta\,\bs e_D + \cos\theta\,\bs e_{D-1}}| \geq \delta$,
with~$\widehat{\bs \varphi} = \tfrac{1}{\|\bs \varphi\|}\bs \varphi$.

We are now ready to integrate to find $p_{1}$:
\begin{align*}
p_{1} =  \mu(\cl V^-_{\delta})\ &= \tinv{(2\pi)^{D/2}}\,\int_{\Rbb_+}\ud
  r\ r^{D-1}e^{-r^2/2}\ \bigg[\big(\int_0^\pi\ud\phi_1\,\sin^{D-2}\phi_1\,\big) \cdots \big(\int_0^\pi\ud\phi_{D-2}\,\sin\phi_{D-2}\,\big)\,\bigg]\cdots\\
  &\qquad\qquad\qquad\qquad
  \big[\int_{[0,\theta]\,\cup\,[\pi,\pi+\theta]}\ud\phi_{D-1}\
  \chi_{g(\delta,\bs\varphi)}(\phi_{D-1})\,\chi_{g(\delta,\bs\varphi)}(\phi_{D-1}-\theta)\,
  \big],
\end{align*}
with $\chi_\lambda(\phi)=1$ if $|\sin\phi\,|\geq \lambda$ and 0 else,
for some $\lambda\in[0,1]$, and
$g(\delta,\bs\varphi)=\delta/(\sin\phi_1\,\cdots\,\sin\phi_{D-2})$.

However,
$$
\int_{[0,\theta]\,\cup\,[\pi,\pi+\theta]}\ud\phi\
\chi_\lambda(\phi)\,\chi_\lambda(\phi-\theta)\ =\ \max(2\theta -
4\arcsin\lambda,\ 0),
$$
and $\max(2\theta - 4\arcsin\lambda, 0) \geq 2\theta -
2\pi\lambda$, since $\lambda\leq \arcsin\lambda \leq
\tfrac{\pi}{2}\lambda$ for any $\lambda\in[0,1]$. Consequently,
\begin{align*}
  \mu(\cl V^-_{\delta})&\geq\ \tinv{(2\pi)^{D/2}}\,\int_{\Rbb_+}\ud
  r\ r^{D-1}e^{-r^2/2}\ \cdots\\
&\quad
  \bigg[\big(\int_0^\pi\ud\phi_1\,\sin^{D-2}\phi_1\,\big) \cdots
  \big(\int_0^\pi\ud\phi_{D-2}\,\sin\phi_{D-2}\,\big)\,\bigg]
  \big(2\theta -
  \tfrac{2\pi\delta}{(\sin\phi_1\,\cdots\,\sin\phi_{D-2})}\big)\\
&=\ \frac{\theta}{\pi}\ -\ 2\pi\delta\ \frac{I_{D-3}\,I_{D-4}\cdots
I_0}{(2I_0)\,I_{1}\cdots I_{D-2}}\ =\ \frac{\theta}{\pi}\ -\ \frac{\pi\,\delta}{I_{D-2}}, 
\end{align*}
with $I_n:=\int_0^\pi\ud\phi \,\sin^{n}\phi$ and knowing that
$(2\pi)^{D/2} = (2 I_0)\,(I_1\,\cdots I_{D-2}) \int_{\Rbb_+}\ud
  r\ r^{D-1}e^{-r^2/2}$ since
$$
\mu(\Rbb^D) = 1 = (2\pi)^{-D/2} (\int_0^{2\pi}
\ud \phi_{D-1})\,(I_1\,\cdots I_{D-2})\,\int_{\Rbb_+}\ud
  r\ r^{D-1}e^{-r^2/2}. 
$$

Using the fact that $I_n =
\sqrt{\pi}\,\Gamma(\tfrac{n+1}{2})/\Gamma(\tfrac{n}{2}+1) \geq
{\sqrt{\pi}}/{\sqrt{\frac{n}{2}+\inv{4}}}$, we obtain $ I_{D-2}\geq
\tfrac{\sqrt{\pi}}{\sqrt{\frac{D}{2}-\frac{3}{4}}}\geq
\sqrt{\tfrac{2\pi}{D}}$, and thus
$$
p_1\ \geq\ \ang{\bs x}{\bs s}\ -\ \sqrt{\tfrac{\pi}{2} D}\,\delta.
$$
If we want a meaningful bound for $p_1\geq 0$, then we must have $\ang{\bs x}{\bs s}\geq
\sqrt{\tfrac{\pi}{2} D}\,\delta \geq \delta$. Therefore, as soon as the lower bound is
positive, the aforementioned condition $\ang{\bs x}{\bs s}\geq\delta$
always holds.

The lower bound for $p_0$ is obtained similarly. It is straightforward
to show that $p_0=\mu(\cl V^+_{\delta})$, with $\cl V^+_{\delta} =
\{\bs \varphi: \scp{\bs\varphi}{\bs x}\scp{\bs\varphi}{\bs s} > 0,
\|\bs x\,-\,\cl P_{\Pi(\varphi)}\,\bs x\| \geq \delta, \|\bs
y\,-\,\cl P_{\Pi(\varphi)}\,\bs s\| \geq \delta\}$. Lower bounding
$\mu(\cl V^+_{\delta})$ as for $\mu(\cl V^+_{\delta})$, the only
difference occurring with the integral on $\phi_{D-2}$ given by 
\begin{multline*}
\int_{[\theta,\pi]\,\cup\,[\pi+\theta,2\pi]}\ud\phi_{D-1}\
\chi_{g(\delta,\bs\varphi)}(\phi_{D-1})\,\chi_{g(\delta,\bs\varphi)}(\phi_{D-1}-\theta)\
\cdots\\ 
=\ 2\pi - 2\theta -
4\arcsin g(\delta,\bs\varphi)\ \geq\ 2(\pi - \theta) - 2\pi g(\delta,\bs\varphi).
\end{multline*}
Therefore, the lower bound of $p_0$ amounts to change $\theta\to
\pi-\theta$ in the one of $p_1$, which provides the result.
\end{proof}

\begin{proof}[Proof of Lemma \ref{lemma-W-V-xy-equiv}] If $\delta = 0$, there is nothing to prove.
  Therefore $\delta > 0$ and if $\bs \varphi^*$ belongs to either
  $\cl V_{\delta}$ or $\cl W_{\delta}$, we must have
  $\scp{\bs\varphi}{\bs x}\scp{\bs\varphi}{\bs s} < 0$.  It is also
  sufficient to work on the restriction of $\cl V_{\delta}$ and
  $\cl W_{\delta}$ to unit vectors.

\noindent{\underline{\emph{(i)} $\cl V_{\delta} \subset
    \cl W_{\delta}$}:} By contradiction, let us assume that $\bs
\varphi^*\in \cl V_{\delta}$ but $\bs \varphi^*\notin
\cl W_{\delta}$. Without any loss of generality, $\scp{\bs
  \varphi^*}{\bs x}> 0$ and $\scp{\bs \varphi^*}{\bs s}<
0$. Since $\bs \varphi^*\notin \cl W_{\delta}$, there exist two
vectors $\bs u^*\in B^*_\delta(\bs x)$ and $\bs v^*\in B^*_\delta(\bs
s)$ such that $\scp{\bs\varphi^*}{\bs u^*}\scp{\bs\varphi^*}{\bs
  v^*}>0$. If $\scp{\bs\varphi^*}{\bs u^*}>0$ and $\scp{\bs\varphi^*}{\bs
  v^*}>0$, then, since $\scp{\bs\varphi^*}{\bs s}< 0$ and by
continuity of the inner product, there exist a $\lambda\in(0,1)$ such
that $\scp{\bs\varphi^*}{\bs s(\lambda)}=0$ with $\bs s(\lambda) = \bs
s + \lambda (\bs v^* - \bs s)$. Therefore, $\bs s(\lambda) \in
\Pi(\bs\varphi)$ and, by definition of the orthogonal projection,
$\|\bs s -\cl P_{\Pi(\bs \varphi)}\,\bs s\| \leq \|\bs s -\bs
s(\lambda)\|\leq \lambda \delta < \delta$ which is a
contradiction.  If $\scp{\bs\varphi^*}{\bs u^*}<0$ and
$\scp{\bs\varphi^*}{\bs v^*}<0$, we apply the same reasoning on $\bs
x$ and $\bs u^*$. Therefore, $\cl V_{\delta} \subset
\cl W_{\delta}$.

\noindent{\underline{\emph{(ii)} $\cl W_{\delta} \subset \cl V_{\delta}$}:}
If $\bs \varphi^*\in \cl W_{\delta}$ with $\bs \varphi^*\notin
\cl V_{\delta}$, we have either $\|\bs x\,-\,\cl
P_{\Pi(\varphi^*)}\,\bs x\| < \delta$ or $\|\bs s\,-\,\cl
P_{\Pi(\varphi^*)}\,\bs s\| < \delta$. Let us say that $\|\bs
x\,-\,\cl P_{\Pi(\varphi^*)}\,\bs x\| < \delta$. Then, for $\bs w =
\bs x\, +\, \delta\, (\cl P_{\Pi(\varphi^*)}\,\bs x - \bs x)/\|\cl
P_{\Pi(\varphi^*)}\,\bs x - \bs x\| \in B^{*}_{\delta}(\bs x)$,
$\scp{\bs \varphi^*}{\bs x}\scp{\bs \varphi^*}{\bs w} = (\scp{\bs
\varphi^*}{\bs x})^2\big(1-\delta/\|\cl
P_{\Pi(\varphi^*)}\,\bs x - \bs x\|\big) + \delta\, \scp{\bs
\varphi^*}{\cl P_{\Pi(\varphi^*)}\,\bs x}< 0$. However,
$\bs \varphi^*\in \cl W_{\delta}$ and
$\scp{\bs \varphi^*}{\bs x}\scp{\bs \varphi^*}{\bs s} < 0$, leading to
$\scp{\bs \varphi^*}{\bs w}\scp{\bs \varphi^*}{\bs s} > 0$, which is a
contradiction.
\end{proof}

\section{Theorem~\ref{thm:1-bit-stable-embed}: Gaussian Matrices Provide B$\epsilon$SEs}
\label{sec:proof-thm-1-bit-stable-embed}
The strategy for proving Theorem~\ref{thm:1-bit-stable-embed} will be
to count the number of pairs of $K$-sparse signals that are Euclidean
distance $\delta$ apart.  We will then apply the concentration results
of Lemma~\ref{lem:conc-prop-hamming-balls} to demonstrate that the
angles between these pairs are approximately preserved.  We
specifically proceed by focusing on a single $K$-dimensional subspace
(intersected with the unit sphere) and then by applying a union bound
to account for all possible subspaces.

Let $T\subset [N]$ be an index set of size $|T|= K$,
$\Sigma^*(T)=\{\bs w\in\Rbb^N: \supp \bs w\subset T,~~\|\bs
w\|_{2}=1\}$ be the sphere of unit vectors with support $T$.  We first
use again the fact that the sphere $\Sigma^*(T)$ can be
$\delta$-covered by a finite set of points $Q_{T,\delta}$. That is,
for any $\bs w\in \Sigma^*(T)$, there exists a $\bs q \in
Q_{T,\delta}$ such that $\bs w \in B^*_\delta(\bs q) = B_\delta(\bs
q)\cap \Sigma^*_T = \{\bs w'\in \Sigma^*_T: \|\bs w'-\bs q\|_{2}\leq
\delta\}$ \cite{BarDavDeV::2008::A-Simple-Proof}. Note that the size of
$Q_{T,\delta}$ is bounded by $| Q_{T,\delta}| \leq C_\delta =
(3/\delta)^K$.

Let $\Phi_{T}$ be the matrix formed by the columns of $\Phi$ indexed
by $T$ and note that $\Phi_{T}\bs w = \Phi\bs w$.  Given $\epsilon'\geq
0$, for all pairs of points $\bs p, \bs q \in
Q_{T,\delta}$, we have
\begin{multline}
\label{bound:union}
  \mathbb{P}\left(\forall\,\bs u\in B^{*}_{\delta}(\bs p),\ \forall \bs v\in B^{*}_{\delta}(\bs q),\ \big|\,d_H\big(A(\bs u),A(\bs v)\big)\ -\
  \ang{\bs p}{\bs q}\,\big|\ \leq\ \epsilon' + \sqrt{\tfrac{\pi}{2} K}\,\delta\, \right)\\ 
  \geq\ 1\ -\ 2\,(\tfrac{3}{\delta})^{2K}\,e^{-2\epsilon'^2 M}.
\end{multline}
This follows from Lemma \ref{lem:conc-prop-hamming-balls} with $D=K$, since $\Phi_T$ is a Gaussian matrix 
and by invoking the union bound, since there are ${C_\delta \choose 2}\leq C^2_\delta=(3/\delta)^{2K}$
such pairs $\bs x, \bs s$. 

The bound (\ref{bound:union}) can be extended to all possible index
sets $T$ of size $K$ via the union bound.  Specifically, for all
$T\subset [N]$ and all pairs of points $\bs p, \bs q \in
Q_{T,\delta}$, we have now jointly
\begin{multline}
\label{eq:allindex}
  \mathbb{P}\left(\forall\bs u\in B^{*}_{\delta}(\bs p), \forall\bs v\in
  B^{*}_{\delta}(\bs q),\ \big|\,d_H\big(A(\bs u),A(\bs v)\big)\ -\
  \ang{\bs
    p}{\bs q}\,\big|\ \leq\ \epsilon' + \sqrt{\tfrac{\pi}{2} K}\,\delta\, \right)\\ 
  \geq\ 1\ -\ 2\,(\tfrac{eN}{K})^K\,(\tfrac{3}{\delta})^{2K}\,e^{-2\epsilon'^2 M}
\end{multline}
since there are no more than ${N \choose K} \leq (eN/K)^K$ possible $T$.

We can reformulate this last result as follows. Let us take any pair
of points on the sphere $\bs x, \bs s\in S^{N-1}$ such that their joint
support $T=\supp(\bs x)\,\cup\,\supp(\bs s)$ has a size $|T|\leq
K$. We have obviously $\bs x, \bs s\in \Sigma^*(T)$.  Taking the
covering set $Q_{T,\delta}$ defined for $\Sigma^*(T)$, there exist two
points $\bs p, \bs q \in Q_{T,\delta}$ such that $\bs x\in
B^{*}_{\delta}(\bs p)$ and $\bs s\in B^{*}_{\delta}(\bs q)$. From
(\ref{eq:allindex}), with a probability exceeding $1 -
2\,(\tfrac{eN}{K})^K\,(\tfrac{3}{\delta})^{2K}\,e^{-2\epsilon'^2 M}$,
we have
\begin{equation}
\label{eq:angleb}
\big|\,d_H\big(A(\bs x),A(\bs s)\big)\ -\ \ang{\bs p}{\bs q}\,\big|\
\leq\ \epsilon' + \sqrt{\tfrac{\pi}{2} K}\,\delta.
\end{equation}
To obtain our final bound, consider that $\bs x\in B^{*}_{\delta}(\bs p)$ implies that $\pi\,\ang{\bs x}{\bs
  p}\leq 2\arcsin \delta/2 \leq \pi\delta/2$, and $\ang{\bs
s}{\bs q}$ can be similarly bounded. Thus, $\ang{\bs x}{\bs s}\geq \ang{\bs p}{\bs q} - \delta$
and $\ang{\bs x}{\bs s}\leq \ang{\bs p}{\bs q} + \delta$, and (\ref{eq:angleb}) becomes
\begin{equation}
\big|\,d_H\big(A(\bs x),A(\bs s)\big)\ -\ \ang{\bs x}{\bs s}\,\big|\
\leq\ \epsilon' + (1+\sqrt{\tfrac{\pi}{2} K})\,\delta.
\end{equation}
Let us define the probability of failure as $2\,(\tfrac{eN}{K})^K\,(\tfrac{3}{\delta})^{2K}\,e^{-2\epsilon'^2
M} = \eta,$ where $0<\eta<1$, and set $\epsilon' =
(1+\sqrt{\tfrac{\pi}{2} K})\,\delta$ and
$2\epsilon'=\epsilon$. Solving for $M$, we finally get that $|d_H(A(\bs x),A(\bs s)) - \ang{\bs x}{\bs s}\,|
\leq \epsilon$ with a probability bigger than $1-\eta$ if
$$
M\ \geq\ \tfrac{2}{\epsilon^2}\big(K\,\log(\tfrac{9eN}{K}) +
2K\,\log(\tfrac{2(1+\sqrt{2\pi K})}{\epsilon}) + \log(\tfrac{2}{\eta}) \big).
$$
Since $K\geq 1$, we have that $2(1+\sqrt{2\pi K}) \leq 2 (1 +
\sqrt{2\pi})\sqrt{K} < 35 \sqrt K/\sqrt{9e}$, and
thus the previous relation is satisfied if 
\begin{align*}
M\ &\geq\ \tfrac{2}{\epsilon^2}\big(K\,\log(\tfrac{9eN}{K}) +
2K\,\log(\tfrac{35/\sqrt{9e}}{\epsilon} \sqrt K) + \log(\tfrac{2}{\eta})
\big),\\
&=\ \tfrac{2}{\epsilon^2}\big(K\,\log(N) +
2K\,\log(\tfrac{35}{\epsilon}) + \log(\tfrac{2}{\eta})
\big).
\end{align*}

\section{Lemma~\ref{prop:noise-corruption}: Stability with Measurement Noise}
\label{sec:proof-lemma-noise-corruption}

In Lemma~\ref{prop:noise-corruption}, since $\Phi\sim\SGR{M}{N}$,
each $y_i = (\Phi\bs x)_i$ follows a Gaussian distribution $\cl
N(0,\|\bs x\|_{2}^2)$, and furthermore, since we have independent
additive noise, $z_i = y_i + n_i = (\Phi\bs x)_i + n_i$ follows the
Gaussian distriubtion $\cl N(0,\|\bs x\|_{2}^2+\sigma^2)$.

We begin by bounding the probability that any noisy measurement
$z_{i}$ has a different sign than the original corresponding
measurement $y_{i}$, \ie we bound $\tilde p := \bb P(z_i y_i < 0)$. This
quantity is interesting since $M\, d_H\big(A_{\bs n}(\bs x),A(\bs
x)\big)$ follows a Binomial distribution with $M$ trials and
probability of success $\tilde  p$ and thus we also have $\bb
E\big(d_H\big(A_{\bs n}(\bs x),A(\bs x)\big)\big) = \tilde  p$.

To solve for the bound, we compute
\begin{align*}
  \tilde  p = \int_{\Rbb} \ud u\ \bb P(z_i y_i < 0\, |\,y_i = u\,)\, f_{y_i}(u)\
  =\ \int_{\Rbb} \ud u\ \bb P(u^2 + u n_i < 0)\,
  g(u;\|\bs x\|_{2}),
\end{align*}
with the pdf $f_{y_i}(t) = g(t;\sigma') = \inv{\sqrt{2\pi}
t}\exp(-t^2/2\sigma'^2)$. This leads to
\begin{align*}
  \tilde  p\ &= \int_{0}^\infty \ud u\ \bb P(n_i < -u)\, g(u;\|\bs x\|_{2})\
  +\ \int_{-\infty}^0 \ud u\ \bb P(n_i > -u)\, g(u;\|\bs x\|_{2})\\
&= \int_{0}^\infty \ud u\ 2\,Q(u/\sigma)\, g(u;\|\bs x\|_{2})\ \leq\ \int_{0}^\infty \ud u\ e^{-\frac{u^2}{2\sigma^2}}\,g(u;\|\bs
x\|_{2})\\ 
&=\ \tinv{\sqrt{2\pi}\|\bs x\|_{2}}\int_{0}^\infty \ud u\
e^{-\tinv{2}(\frac{(\|\bs x\|_{2}^2 + \sigma^2)\, u^2}{\sigma^2\|\bs
x\|_{2}^2})}\ =\ \tinv{2} \frac{\sigma}{\sqrt{\|\bs x\|_{2}^2 + \sigma^2}},
\end{align*}   
where $Q(u)=\int_u^\infty\ud t\, g(t;1)$ denotes the tail integral of the standard Gaussian
distribution which is bounded by $Q(t)\leq\tinv{2}e^{-t^2/2}$ for
$t\geq 0$
(see for instance \cite[Eq. (13.48)]{JKB-CUDV1}). 

Thus, we have $\tilde  p \leq e(\sigma,\|\bs x\|_{2})=\tinv{2} \frac{\sigma}{\sqrt{\|\bs x\|_{2}^2
+ \sigma^2}}$ and, by applying the Chernoff-Hoeffding inequality to the distribution of $d_H\big(A_{\bs n}(\bs x),A(\bs x)\big)$, 
\begin{eqnarray*}
\lefteqn{\bb P\big[M\,d_H\big(A_{\bs n}(\bs x),A(\bs x)\big) > \ M\,e(\sigma,\|\bs x\|_{2}) +
M\epsilon\,\big]}\hspace{3cm}&&\\
&\leq& \bb P\big[M\,d_H\big(A_{\bs n}(\bs x),A(\bs x)\big) >\ M\,\tilde  p + M\epsilon\,\big]\\ 
&\leq&\ e^{-2M\epsilon^2},
\end{eqnarray*}
which proves the lemma.

\section{Asymptotic bound on $\epsilon$ in Theorems \ref{thm:hash} and
\ref{thm:1-bit-stable-embed}}
\label{app:proof_epsilon_MK_complex}
Both Theorem \ref{thm:hash} and Theorem \ref{thm:1-bit-stable-embed}
provide guarantees on the worst-case error $\epsilon$ of the form
\begin{equation}
  \label{eq:gener-eps-def}
  \epsilon^n \le \tfrac{1}{M}\,\big( \alpha\,K\log(N) +
  \beta\log(\tinv{\eta}) + \gamma\,K \log(\tfrac{1}{\epsilon})\,\big),
\end{equation}
for some exponent $n\in\{1,2\}$, and for given constants $\alpha,\beta,\gamma> 0$.

In this appendix we show that, considering $0<\eta<1$ fixed, the relation
(\ref{eq:gener-eps-def}) implies
\begin{equation}
  \label{eq:gener-eps-asymp}
  \epsilon^n = O\big( \tfrac{K}{M} \log (\tfrac{MN}{K}) \big)
\end{equation}
asymptotically in ${M}/{K}$ and $N$. Notice that, up to a redefinition
$\epsilon^n \to \epsilon$ and $\gamma/n \to \gamma$, it is sufficient
to prove the relation for $n=1$. We also define $\rho := \beta
\log(1/\eta)$.

First, we consider $N$ fixed and show $\epsilon = O\big( \tfrac{K}{M}
 \log (\tfrac{M}{K}) \big)$.
Let us assume this is not the case, i.e., for all
$c>0$, and all $R_0 > 0$, there exists a ratio ${M}/{K} > R_0$ such
that $\epsilon > c\,({K}/{M}) \log ({M}/{K})$. Therefore,
$$
\log \tinv{\epsilon} < \log \tfrac{M}{K} - \log\big( c \log
\tfrac{M}{K}\big) < \log \tfrac{M}{K} - \log( c \log R_0),
$$
Thus~\eqref{eq:gener-eps-def} becomes
$$
\epsilon \leq \tfrac{1}{M}\,\big( \alpha\,K \log N + \rho + \gamma\,K
\log \tfrac{M}{K} - \gamma\,K\,\log( c \log R_0)).
$$

Using $\epsilon > c\,({K}/{M}) \log ({M}/{K})$, this last inequality
becomes
\begin{align}
  \alpha\,\log N + \tinv{K}\,\rho + \gamma\,\log \tfrac{M}{K}
  \ >\ c\, \log \tfrac{M}{K} + \gamma\,\log( c \log R_0).
  \label{eq:ineq_error}
\end{align}

For fixed $N$ and $\eta$, and since we reasonably have $K\geq 1$, the
parameters $c$ and $R_0$ can always be selected so that $ \gamma \log(
c \log R_0) > \alpha\,\log N + \rho/K$.  In this
case,~\eqref{eq:ineq_error} implies $\gamma \log\tfrac{M}{K} >
c\,\log \tfrac{M}{K}$. Taking $c> \gamma$, which is still compatible
with the selection of $R_0$ and $c$ above, leads to a
contradiction. Thus $\epsilon=O\left(\tfrac{K}{M}
\log\tfrac{M}{K}\right)$ for fixed $N$.

Next, we assume $N$ varies and $R:=M/K$ is fixed, and show that
$\epsilon^n = O((1/R) \log (R N))$. We again restrict the analysis to
$n=1$. Now we assume for all $c>0$ and all $N_0>0$, there is an $N > N_0$ such that $\epsilon > (c/R) \log(R N)$. Therefore, $\log
\tfrac{1}{\epsilon} < - \log ( (c/R) \log (R N_0))$ and
(\ref{eq:gener-eps-def}) becomes
$$
\tfrac{c}{R} \log(R N) + \tfrac{\gamma}{R} \log ( \tfrac{c}{R} \log (RN_0)) < \tfrac{\alpha}{R} \log N + \tfrac{1}{M}\rho.   
$$ 
Since $R$ and $\eta$ are fixed and we have $M\ge 1$, the parameters $c$ and $N_0$ can be selected so that $ \tfrac{\gamma}{R} \log ( \tfrac{c}{R} \log (RN_0)) >\tfrac{\rho}{M}$, that implies $\tfrac{c}{R} \log(R N) < \tfrac{\alpha}{R} \log N$. Taking $c >\tfrac{\alpha}{\log R} $ leads to a contradiction and completes the proof.

\end{document}